\documentclass[letterpaper,twocolumn,10pt]{article}
\usepackage{jsys_camera_ready}

\microtypecontext{spacing=nonfrench}
\usepackage{algorithm, algpseudocode}
\usepackage{soul}
\usepackage{diagbox}

\usepackage{epsfig, endnotes, xurl, color, graphicx}
\usepackage{mathtools, amsmath, amsfonts}
\usepackage{diagbox, booktabs, colortbl}
\usepackage{listings, framed}

\usepackage{orcidlink}
\usepackage{adjustbox}
\makeatletter
\let\MYcaption\@makecaption
\makeatother

\usepackage{silence} 
\usepackage[font=footnotesize]{subcaption}

\makeatletter
\let\@makecaption\MYcaption
\makeatother

\usepackage{multirow}

\usepackage{xspace, slashbox, enumitem, bm}
\usepackage{xcolor}
\usepackage{derivative}
\usepackage{tabularx}
\usepackage{threeparttable}
\DeclareMathAlphabet{\mathcal}{OMS}{cmsy}{m}{n}
\usepackage[mathscr]{eucal}
\usepackage{float}
\usepackage{breqn}
\usepackage[]{hyperref}
\usepackage{cleveref}
\usepackage{amsthm}
\newtheorem{definition}{Definition}[section]
\newtheorem{theorem}{Theorem}[section]
\newtheorem{lemma}[theorem]{Lemma}
\newtheorem{corollary}[theorem]{Corollary}

\mathchardef\mhyphen="2D % Define a "math hyphen"

\microtypecontext{spacing=nonfrench}
\newcommand\sysone{chainBoost\xspace}
\newcommand\systwo{marketPrism\xspace}

\usepackage{ifthen}
\usepackage[normalem]{ulem} % for \sout
\usepackage{xcolor}
\usepackage{amssymb}

\newboolean{showedits}
\setboolean{showedits}{true} % toggle to show or hide edits
\ifthenelse{\boolean{showedits}}
{
	\newcommand{\del}[1]{\textcolor{red}{\sout{#1}}} % please delete
}{
	\newcommand{\del}[1]{} % please delete
	
}

\newboolean{showcomments}
\setboolean{showcomments}{true}
\newcommand{\id}[1]{$-$Id: scgPaper.tex 32478 2010-04-29 09:11:32Z oscar $-$}

\ifthenelse{\boolean{showcomments}}
{\newcommand{\nbc}[3]{
 {\colorbox{#3}{\bfseries\sffamily\scriptsize\textcolor{white}{#1}}}
 {\textcolor{#3}{\sf\small$\blacktriangleright$\textit{#2}$\blacktriangleleft$}}}
 }
{\newcommand{\nbc}[3]{}
 \renewcommand{\del}[1]{} % please delete
 }

\definecolor{jccolor}{rgb}{0.2,0.4,0.6}
\definecolor{clcolour}{rgb}{0.5,0.7,0.9}

\usepackage{wasysym}

\newcommand\ghada[1]{\nbc{GA}{#1}{magenta}}
\newcommand\mo[1]{\nbc{MN}{#1}{teal}}

\newcommand\sk{\mathsf{sk}\xspace}
\newcommand\pk{\mathsf{pk}\xspace}

\newcommand\cs{S_c\xspace}

\newcommand\com{\mathcal{C}\xspace}

\newcommand\led{\mathcal{L}\xspace}
\newcommand\summ{\mathsf{sum}\xspace}
\newcommand\cid{\mathsf{cid}\xspace}
\newcommand\amount{\mathsf{amount}\xspace}
\newcommand\clients{\mathcal{C}\xspace}
\newcommand\servers{\mathcal{S}\xspace}
\newcommand\miners{\mathcal{M}\xspace}
\newcommand{\param}{\ensuremath{\mathtt{pp}}\xspace}
\newcommand\syssetup{\mathsf{SystemSetup}\xspace}
\newcommand\partysetup{\mathsf{PartySetup}\xspace}
\newcommand\stt{\mathsf{state}\xspace}
\newcommand\createTransaction{\mathsf{CreateTx}\xspace}
\newcommand\tx{\mathsf{tx}\xspace}
\newcommand\ask{\mathsf{ask}\xspace}
\newcommand\offer{\mathsf{offer}\xspace}
\newcommand\agreement{\mathsf{agreement}\xspace}
\newcommand\transfer{\mathsf{transfer}\xspace}
\newcommand\dispute{\mathsf{dispute}\xspace}
\newcommand\sync{\mathsf{sync}\xspace}
\newcommand\serviceProof{\mathsf{serviceProof}\xspace}
\newcommand\servicePayment{\mathsf{servicePayment}\xspace}
\newcommand\aux{\mathsf{aux}\xspace}
\newcommand\verifyTransaction{\mathsf{VerifyTx}\xspace}
\newcommand\updateState{\mathsf{UpdateState}\xspace}
\newcommand\mainc{\mathsf{mc}\xspace}
\newcommand\sidec{\mathsf{sc}\xspace}
\newcommand\leader{\mathsf{leader}\xspace}
\newcommand\elect{\mathsf{Elect}\xspace}
\newcommand\createSyncTransaction{\mathsf{CreateSyncTx}\xspace}
\newcommand\verifySyncTransaction{\mathsf{VerifySyncTx}\xspace}
\newcommand\verifyBlock{\mathsf{VerifyBlock}\xspace}
\newcommand\btype{\mathsf{btype}\xspace}
\newcommand\setup{\mathsf{Setup}\xspace}
\newcommand\meta{\mathsf{meta}\xspace}
\newcommand\summary{\mathsf{summary}\xspace}
\newcommand\txtype{\mathsf{txtype}\xspace}

\newcommand\block{\mathsf{B}\xspace}
\newcommand\prune{\mathsf{Prune}\xspace}
\newcommand\sidechain{\mathsf{sc}\xspace}

\newcommand\module{\mathsf{module}\xspace}
\newcommand\modules{\mathsf{modules}\xspace}
\newcommand{\service}{\mathsf{servicePaymentExchange\xspace} }
\newcommand{\marketmatch}{\mathsf{marketMatch\xspace}}
\newcommand\countt{\mathsf{count}\xspace}
\newcommand\payment{\mathsf{payment}\xspace}

\newcommand{\ifcond}[1]{\textbf{if} {#1} \textbf{then}}
\newcommand{\forloop}[1]{\textbf{for} {#1} \textbf{do}}

\newcommand{\elsecond}{\textbf{else}}

\newcommand{\fitmath}[1]{\adjustbox{max width=\linewidth}{$\displaystyle #1$}}
\newcommand{\EE}{\mathbb{E}}

\begin{document}
\date{}

\title{Hierarchical Functionality-oriented Layer-2 Scaling for Decentralized Resource Markets} 
\author{
  {\textsc{Mohamed E. Najd}}\orcidlink{0009-0009-0890-7604}\\
  University of Connecticut\\
  menajd@uconn.edu
  \and
  {\textsc{Ghada Almashaqbeh}}\orcidlink{0000-0003-1297-0760}\\
  University of Connecticut\\
  ghada@uconn.edu
  % copy the following lines to add more authors \and {\rm Name}\\
  % Name Institution
} % end author

% make the title area
\maketitle

\thispagestyle{empty}

\begin{abstract}
Decentralized resource markets are Web 3.0 applications
that build open-access platforms for trading digital resources. They promise cost reduction, transparency, and flexible service provision. However, these markets suffer from scalability problems due to their large workloads that must be processed in a timely manner.

We introduce \systwo, a secure hybrid sidechain-sharding layer-2 solution for decentralized resource markets. \systwo enables effective scaling by accounting for the work models and traffic patterns of these markets. It parallelizes traffic processing by leveraging \emph{dependent sidechains} and \emph{functionality-oriented workload splitting} that assigns each market module its own sidechain. Different from sharding, \systwo \emph{does not} incur any cross-sidechain transactions that tend to be costly. Moreover, it introduces \textit{hierarchical workload sharing} that further sub-divides overloaded modules, and \textit{weighted miner assignment} that assigns miners with vested interest to critical modules' sidechains to promote resilience. \systwo also enables blockchain pruning to reduce the market storage footprint. Besides analyzing security, we build a proof-of-concept implementation for a distributed file storage market as a use case. Our experiments show that, compared to a single sidechain-based prior solution, \systwo boosts throughput by 4x and reduces confirmation latency by 5x. Also, they show that \systwo outperforms sharding by 2.5x in throughput and 3.5x in latency. 
\end{abstract}

%\begin{IEEEkeywords}
%Blockchain-based resource markets, layer-2 scalability, hybrid sidechain-sharding 
%\end{IEEEkeywords}

\section{Introduction}
\label{intro}
Decentralized resource markets are Web 3.0 applications that build open-access platforms for trading digital resources, e.g., computation outsourcing~\cite{golem}, file storage~\cite{filecoin}, and video transcoding~\cite{livepeer}. They aim to solve the cost, trust, and governance issues of centrally-managed services. For example, conventional content distribution network (CDNs) can offload traffic during peak periods to decentralized CDNs to improve performance~\cite{Anjum17}. Moreover, they mitigate information loss due to provider misconfiguration/availability issues, e.g., Google Drive users lost their data even under Google's multi-region infrastructure~\cite{tomshardwareGoogleDrive}. Furthermore, these markets resolve incentive issues of early P2P systems that suffered from freeloading~\cite{feldman2004free} or required centralized banks to handle payments~\cite{yang2003ppay}. In parallel, they offer insights into useful mining by tying miners' power to service contribution.

However, the high potential of these markets is impeded by several limitations. Facilitating collaboration between trustless parties in a monetary-incentivized and open-access model, introduces various design complexities and security issues. Ensuring correct and secure operation usually requires deploying new techniques, such as market management, resource expenditure proofs, and dispute resolution, which negatively impact performance. At the same time, digital services must meet particular QoS levels in terms of throughput and latency. Under the large workload they have, having a blockchain makes it harder to meet such requirements due to scalability issues; e.g., Ethereum, a popular choice for Web 3.0 applications, has an average throughput of 14
transactions/sec~\cite{eth-throughput} with around 62\% of the transactions have a delay of at least 5 min~\cite{lao2024unpacking}. All these factors hinder the vision of Web 3.0 that aims to replace the centrally-managed services that we currently use with fully-decentralized ones. 

\begin{table*}[ht!]
\caption{Comparison with prior work.}
%\vspace{-6pt}
\label{tbl:comparison}
\resizebox{\textwidth}{!}{
\begin{threeparttable}

\begin{tabular}{lcccccccccc}
\toprule
Solution\tnote{a}  & Type & Finality  &\begin{tabular}[c]{@{}c@{}} Workload \\ Distribution \end{tabular} &\begin{tabular}[c]{@{}c@{}} Miner Population \\ Based Scaling  \end{tabular} & \begin{tabular}[c]{@{}c@{}}  Intra-division \\ Scaling \end{tabular} & \begin{tabular}[c]{@{}c@{}}  Cross-division \\ Tx Rate\tnote{b} \end{tabular}& Pruning & On-chain Storage \\
\midrule
Optimism\cite{Optimism, optimismWithdrawalFlow} & Optimistic rollup & 7 days & N/A          & No                 & No                 & N/A        & No & \begin{tabular}[c]{@{}c@{}}    Batch transcript \\ and state changes  \end{tabular}\\
\midrule
ZKSync Era\cite{zksync, zksync-era-finality}  & ZK-rollup & 3-24h & N/A         & No                 & No                 & N/A        & Yes & State changes    \\
\midrule
Omniledger\cite{Kokoris18} & Random sharding  & \begin{tabular}[c]{@{}c@{}}   1 block \\ (order of seconds) \end{tabular} & Random & Yes & No  & 90\%+ & Optional  & \begin{tabular}[c]{@{}c@{}}    All transactions  \\ / State blocks\tnote{c}  \end{tabular}\\
\midrule
Rapidchain\cite{zamani2018} & Random sharding  & \begin{tabular}[c]{@{}c@{}}   1 block \\ (order of seconds) \end{tabular}   & Random & Yes & No  & 90\%+ & No  & All transactions  \\
\midrule
OptChain\cite{nguyen2019optchain}  & Localized sharding & Blockchain-dependent\tnote{d}    & Score-based       & Yes                 & No                 & 9.28\%        & No & All transactions               \\
\midrule
\sysone\cite{chainboost-paper}  & Dependent sidechain  &    \begin{tabular}[c]{@{}c@{}}   1 meta-block \\ (order of seconds) \end{tabular}    & None    & No &     No
&     N/A             & Yes      & \begin{tabular}[c]{@{}c@{}}  State changes (mainchain) \\ summary-blocks (sidechain) \end{tabular}            \\
\midrule
\systwo      &  Hybrid sidechain-sharding &  \begin{tabular}[c]{@{}c@{}}   1 meta-block \\ (order of seconds) \end{tabular}           & Functionality-based     & Yes                 & Yes                 & 0 \% \tnote{e}             & Yes      & Same as \sysone     \\
\bottomrule
\end{tabular}

\begin{tablenotes}
\item[a] Finality is the time needed to consider state changes final. Intra-division scaling is the ability to scale within a rollup/shard/sidechain. For \sysone, workload distribution is only about sending all service-related traffic to a single sidechain. Optchain builds a UTXO transaction graph, with transactions that have the same PageRank score assigned to the same shard. For sharding solutions, cross-sharding rates are reported for four shards.

\item[b] Cross-division transactions happen when their state effects cannot be committed by a single chain's consensus decision.

\item[c] When pruning is not used, all transactions are logged in the shards. Otherwise, state blocks containing Merkle tree roots of all transactions are stored in the shard, while clients store Merkle proofs.

\item[d] Optchain's finality is dependent on the consensus algorithm used in the shards.

\item[e] Zero by construction under a self-contained partition. Inter-module state reads remain, but use a miner's local view of the system.

\end{tablenotes}

\end{threeparttable}

}

\end{table*}

\subsection{Limitations of Prior Work}
Layer-2 solutions are a keystone of blockchain scalability. They boost throughput and reduce delays by processing transactions off-chain and recording only state changes on-chain. Principal layer-2 paradigms that allow processing arbitrary transactions, not just currency transfers as in payment networks~\cite{Decker15}, generally include rollups and sidechains. Applying these paradigms to resource markets will exacerbate performance and security issues; they do not account for the unique design, QoS requirements, and traffic patterns of these markets, as we discuss below (a comparison can be found in Table~\ref{tbl:comparison}). We also discuss sharding as \systwo utilizes a sharding-like approach at the layer-2 level.\medskip  

\noindent\textbf{Rollups.} Optimistic rollups~\cite{kalodner2018arbitrum,Optimism} process transaction batches off-chain while submitting results on-chain. Verifiers check these results and challenge incorrect ones. Optimistic rollups have long contestation periods, reaching one week~\cite{Optimism}, thus severely impact performance since results are not final until this period is over. Furthermore, they may use trusted verifiers~\cite{op-verifier,arbitrumStateArbitrums}, while incentivizing untrusted verifiers is an open question~\cite{li2023security, mediumOptimisticRollup}---meaning that verifiers are not incentivized to vet/challenge the results, leading to adopting invalid state changes. They also suffer from DoS attacks against fraud-proofs (challenging incorrect results) that prevent their processing during the contestation period~\cite{koegl2023attacks}.  

Zero-knowledge (ZK) rollups~\cite{zksync,liu2024pianist} generate ZK proofs attesting for batch processing correctness. Proof generation is costly, takes on the order of minutes, which becomes worse when attesting to complex transactions~\cite{celer2024zk}. ZK-rollups also have high confirmation delays that may reach up to 24 hours~\cite{zksync,zknationZIP4Reduce}, and require high-performance hardware---operating a node can cost 1365 USD per month~\cite{chaliasos2024analyzing}.

Data availability in rollups requires recording the full transaction batch on-chain, thus rollups do not reduce the mainchain storage cost.\footnote{EIP-4844~\cite{eip-4844} in Ethereum aims to place transaction batches in blobs that are pruned after a few weeks once the contestation period is over.} Some ZK-rollups avoid that at the expense of more expensive ZK proofs~\cite{chaliasos2024analyzing, zksyncDataAvailability}. All these factors make rollups unsuitable for resource markets as they impact QoS support.\medskip

\noindent\textbf{Sidechains.} Sidechains~\cite{Back14, Gavzi19,Kiayias19,Connor17,lee2021hierarchical, gai2021cumulus} operate a secondary blockchain in parallel to the mainchain. Despite their potential in improving scalability, most prior works focused on currency transfer, or two-way peg, between the side and main chains. Moreover, all of them (including industrial deployments, e.g., Cosmos~\cite{cosmos} and parachains such as Polkadot~\cite{polkadot}) employ independent sidechains, where each chain has its own miners, transactions, and network protocol. This limits the utility of sidechains and prevents workload sharing and arbitrary data exchange between the two chains. Cosmos~\cite{cosmos} allows some form of data exchange using data duplicates, but this amplifies storage costs. Furthermore, none of these sidechain solutions support blockchain pruning.

\emph{Dependent sidechains.} To address these limitations,~\cite{chainboost-paper} proposed a new sidechain architecture, called \sysone, that has mutual-dependency relationship with the mainchain. All service-related traffic is offloaded to the sidechain, which processes this traffic (into temporary blocks) and summarize it to produce concise state changes. These summaries are used to sync the mainchain, so once syncing is confirmed, the temporary blocks are pruned. By sharing the workload with the sidechain, \sysone significantly boosts throughput, and reduces latency and blockchain size. However, it operates a \emph{single sidechain}; under heavy market workload, this sidechain will be overwhelmed. Also, \sysone \emph{does not scale with the number of miners}; for every epoch a committee of the mainchain miners manages the sidechain while the rest of the miners do not contribute to workload processing.\medskip

\noindent\textbf{Sharding.} Sharding aims to achieve parallelism by scaling with the number of miners. It splits the blockchain into shards and distributes the miners and system workload among these shards to enable parallel processing. Earlier sharding schemes are either centralized~\cite{Danezis16}, have fixed shard membership~\cite{wang2019mono}, target permissioned blockchains~\cite{zheng2021meepo}, or suffer from security issues---shards are assigned small miner committees that can be taken over by the adversary~\cite{Luu16}. The works~\cite{Kokoris18, zamani2018} address these issues via decentralized 
large enough shard assignment. However, they randomly assign transactions to shards leading to a high volume of cross-shard transactions (more than 90\% of the transactions~\cite{mizrahi2020blockchain, Kokoris18,zamani2018,nguyen2019optchain}). Cross-shard transactions are complex, introduce long delays, and often require end users to handle them~\cite{Kokoris18}, making these solutions unsuitable for resource markets. Solutions that improve transaction distribution locality~\cite{nguyen2019optchain,li2022achieving, tao2024throughput,gao2022pshard,hong2021pyramid} create imbalanced shards, with overloaded shards become performance bottlenecks. Lastly,~\cite{liu2024dynashard} discuss workload-based sub-sharding, but the design is high level and requires a highly-interactive protocol to handle conflicted transactions.\medskip

\noindent\textbf{An open question.} Thus, we ask the following question: \emph{Can we build a layer-2 scalability framework that enables parallel processing in resource markets while accounting for their design and traffic patterns, and without impacting the correct and secure operation of these markets?}

\subsection{Our Contributions}
We answer this question in the affirmative and introduce \systwo, a hybrid sharding-sidechain layer-2 solution. \systwo divides the market into operational modules, each of which is assigned to a sidechain, and distributes the miners and the service-related workload among these sidechains. Thus, it simplifies scaling coordination and eliminates expensive cross-sidechain transactions.\medskip

\noindent\textbf{System design.} We introduce a novel \textit{hybrid sharding-sidechain architecture} that enables parallel processing of the resource market's workload. It introduces an \emph{adaptive functionality-oriented} approach; the resource market is divided into modules, each of which represents a functional unit and is assigned its own sidechain. Thus, the market workload is distributed among these sidechains based on the functionality they handle. For example, all transactions related to service agreements are assigned to the market matching module's sidechain, and all transactions related to service delivery and payments are assigned to the service-payment exchange module's sidechain, and so on. Moreover, we introduce a \textit{hierarchical intra-module workload sharing} that subdivides overloaded modules, by having sub-sidechains, allowing for parallel processing even within the same module. 

With this functionality-oriented workload splitting, transactions are contained within their modules' sidechains, thus \emph{eliminating cross-sidechain transactions}, which significantly boosts performance. In turn, the miners are distributed among these sidechains/sub-sidechains, allowing for scaling further as the number of miners increases in the system.

\systwo retains the underlying features of the sidechain architecture of~\cite{chainboost-paper}. These include managing the sidechain using a practical Byzantine fault-tolerant (PBFT)-based consensus to speed up agreement and reduce confirmation delays; maintaining the mainchain as the source for collective system state via sidechain syncing; and sidechain pruning that significantly reduces the size of the sidechain and subsequently the mainchain. Furthermore, a transaction is considered final once it appears on a sidechain, and workload processing correctness (based on the original logic of the market) is guaranteed by the security of consensus.

To promote resilience to security threats, we introduce a \textit{weighted miner assignment} that assigns miners with high positive participation to sidechains managing critical modules. This reduces the chances of interruptions, caused by miner misbehavior, in these modules. Additionally, \systwo is equipped with an \textit{autorecovery protocol} that accounts for impacted modules during interruptions to allow safe recovery.\medskip

\noindent\textbf{Security analysis.} We analyze the security of \systwo, showing that it preserves the correct and secure operation (i.e., liveness and safety) of the underlying resource market.\medskip

\noindent\textbf{Implementation and evaluation.} To assess performance, we build a proof-of-concept implementation, and conduct various experiments for a file storage market as a use case. Our results show that \systwo achieves up to 4x increase in throughput and 5x reduction in confirmation latency compared to \sysone. It also achieves a 2.5x improvement in throughput, 3.5x reduction in confirmation latency, and a 3x reduction in storage size compared to sharding.

To the best of our knowledge, \systwo is the first to combine the parallelism of sharding and the features of dependent sidechains while achieving the best of both worlds. Coupled with its robustness and interruption resilience, our work is a leap forward towards making decentralized resource markets, and potentially other types of Web 3.0 systems, practically viable without compromising security.

\section{Background}
\label{sec:background}

\noindent\textbf{Decentralized resource markets.} 
These are blockchain systems that offer digital services on top of the currency exchange medium. Users who can provide the service join as servers, while those who need the service join as clients, so resources are being traded among them for cryptocurrency payments. There are numerous deployments in practice, such as Filecoin~\cite{filecoin} for file storage, Livepeer for video transcoding~\cite{livepeer}, and Session~\cite{session} for instant messaging.

A decentralized resource market consists of several modules: market matching, service-payment exchange, and dispute resolution. Servers submit offers detailing the service parameters they can provide, e.g., amount, duration, and prices, while clients submit asks detailing their service needs. Based on these offerings/asks, clients and servers are matched (this may include negotiating service terms). After that, the service-payment exchange phase starts, during which servers serve clients and clients pay for the delivered service. To avoid collecting payment without doing the promised work, servers submit service delivery proofs in order to be paid, e.g., proof of storage~\cite{Fisch19}. Also, to protect against clients who do not pay after being served, clients deposit payments in advance in an escrow that the miners use to dispense payments to the servers based on the submitted valid proofs.

Several security threats may arise due to having monetary incentives, where many of these misbehaviors cannot be mitigated in a cryptographic way. Instead, financial security countermeasures are employed via a dispute resolution module. Upon detecting a misbehavior, a proof-of-cheating is submitted. If this proof is found to be valid, financial penalties will imposed against the misbehaving parties, e.g., forfeiting (all or part of) their collateral or stake.\medskip

\noindent\textbf{The \sysone framework.} \sysone~\cite{chainboost-paper} introduces a new sidechain architecture targeting decentralized resource markets to boost their throughput, and reduce confirmation delays and blockchain size. It enables workload sharing and arbitrary data exchange between the mainchain and the sidechain, thus creating a mutual-dependency relationship between them. That is, these chains live in the same domain---they have the same network protocol, transaction format, and miner population. Also, their security and correct operation depend on each other.

All service-related transactions that can be summarized are processed on the sidechain, while the rest stay on the mainchain. The system operates in epochs and rounds; an epoch is $\omega$ contiguous rounds, and a round is the time during which a block is produced. For each epoch, a committee from the mainchain miners is elected to manage the sidechain. This sidechain is composed of temporary meta-blocks that contain processed transactions, and permanent summary-blocks that contain summaries of the state changes induced by meta-blocks. In every round in an epoch, the committee produces a meta-block, while in the last round of this epoch, it produces a summary-block. Based on the summary-block, the committee issues a sync-transaction to sync the mainchain according to these summaries. Once this transaction is confirmed on the mainchain, all corresponding meta-blocks in that epoch are discarded, allowing for a significant reduction of blockchain size. The permanent summary-blocks support public verifiability as they attest to the validity of state changes.

\begin{figure*}[ht!]
\begin{center}
\includegraphics[height= 2.0 in, width = 1.9\columnwidth]{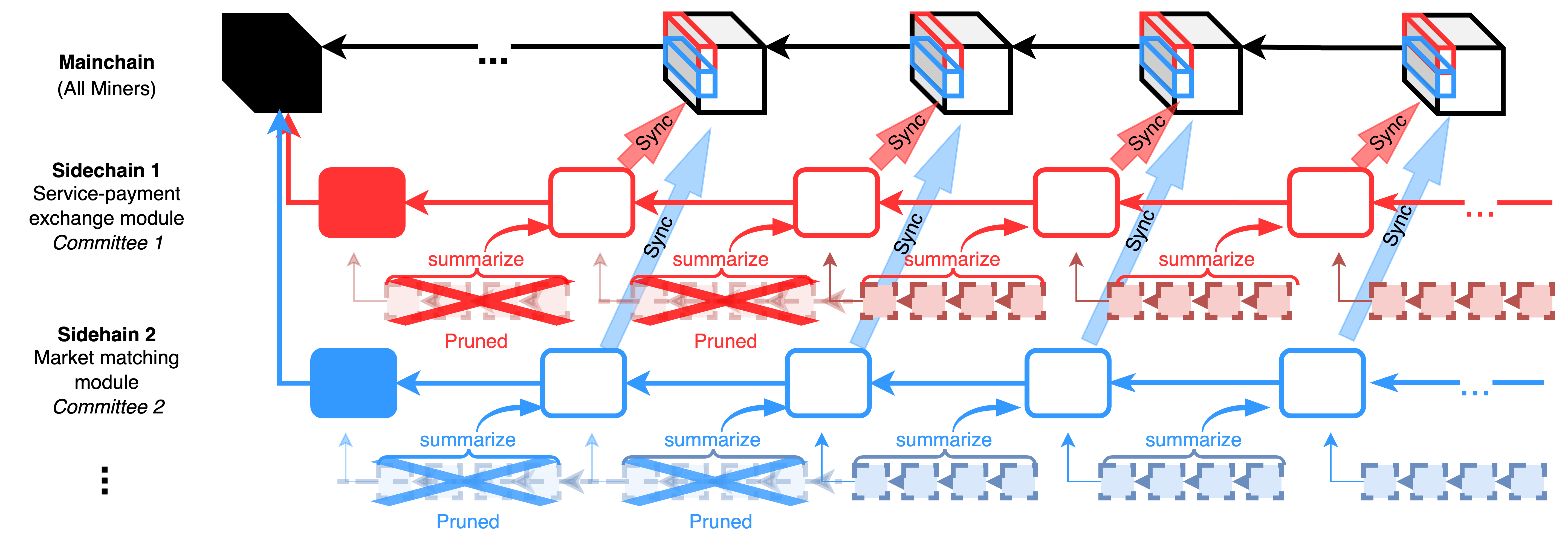}
\end{center}
\vspace{-10pt}
\caption{\systwo architecture.}
\label{chainscale-diag}
\end{figure*}

In \sysone, the sidechain employs a practical Byzantine fault tolerance (PBFT)-based consensus to speed up agreement, and thus allow processing service-related traffic in a timely manner. Any secure PBFT-based consensus with a secure committee election mechanism can be used. Hence, a transaction is final once it appears in a meta-block. As shown in~\cite{chainboost-paper}, \sysone achieves significant performance gains in terms of throughput, latency, and blockchain size. Nonetheless, it operates only one sidechain, and so all service-related transactions is directed to this sidechain, which would lead to a performance bottleneck as explained earlier.

\section{Preliminaries}
\label{sec:prelim}
In this section, we present our notation, and the various models underlying our system design. We adapt the system, security, and adversary models from~\cite{chainboost-paper} to capture \systwo setting, as shown below.\medskip

\noindent\textbf{Notation.} We use $\lambda$ for the security parameter, $\param$ for the system public parameters, $\tx$ for a transaction, and $\led$ for a ledger (or blockchain), with $\led_{\mainc}$ denotes the mainchain ledger and $\led_{\sidec, \module}$ denotes the ledger of the sidechain assigned to the market module $\module$. The set of market modules is denoted as $\modules$, and so we write $\module \in \modules$ to refer to a module within this set. Each participant maintains a secret and public key, $\sk$ and $\pk$, respectively.\medskip

\noindent\textbf{System model.} A resource market is an open-access system that allows anyone to join and leave at any time, and is powered by a permissionless public blockchain. Parties are identified by their public keys, and are divided into clients $\clients$, servers $\servers$, and miners $\miners$ who manage the blockchain and establish Sybil-resistant identities based on their mining power.  We present a system model for a minimum viable, extensible, resource market composed of three modules: market matching, service-payment exchange, and dispute resolution. This market offers functionalities related to system and party setup, transaction issuance and verification, and chain extension.

\systwo introduces several sidechains, each representing a resource market module, and the service-related transactions are split among these sidechains based on the modules they manage. Each sidechain is managed by an epoch committee elected from the mainchain miners that runs a PBFT-based consensus protocol. To simplify the presentation, we adopt a leader-based PBFT in which a leader proposes a block for the committee to agree on (as in~\cite{Kogias16})---nonetheless, voting-based PBFT (as in~\cite{Gilad17}) can be used instead. 

The resource market offers the following functionalities:
\begin{description}
    \item[$\syssetup(1^{\lambda}) \rightarrow (\param, \led^0_{\mainc})$:] Takes the security parameter $\lambda$ as input, and outputs the system public parameters $\param$ and an initial ledger state $\led^0_{\mainc}$ (which is the genesis block).\footnote{For the rest of the algorithms, if not mentioned, the input $\param$ is implicit.}

    \item[$\partysetup(\param) \rightarrow (\stt)$:] Takes the system public parameters $\param$ as input, and outputs the initial state of the party $\stt$ (containing a keypair $(\sk, \pk)$, and for miners, the current ledger state $\led_{\mainc}$).

    \item[$\createTransaction(\txtype, \aux) \rightarrow (\tx)$:] Takes as input the transaction type $\txtype$ and auxiliary information/inputs $\aux$, and outputs a transaction $\tx$ of one of the following types:
    \begin{itemize}\itemsep0pt
        \item $\tx_{\ask}$: Used by a client to state its service needs.
        \item $\tx_{\offer}$: Used by a server to state its service offering.
        \item $\tx_{\agreement}$: Service agreement between a client and a server (or a set of servers).
        \item $\tx_{\serviceProof}$: Service delivery proof submitted by a server.
        \item $\tx_{\servicePayment}$: Server payment for the provided service.
        \item $\tx_{\dispute}$: Initiates a dispute for a misbehavior.
        \item $\tx_{\transfer}$: Currency transfer between participants.
    \end{itemize}

    \item[$\verifyTransaction(\tx) \rightarrow (0/1)$:] Takes transaction $\tx$ as input, and outputs 1 if $\tx$ is valid based on its type semantics and syntax, and 0 otherwise.

    \item[$\verifyBlock(\led_{\mainc}, \block) \rightarrow (0/1)$:] Takes as input the current mainchain ledger state $\led_{\mainc}$ and a new block $\block$, if the block is valid based on the semantics/syntax/current state of the chain, it outputs 1, otherwise, it outputs 0.

    \item[$\updateState(\led_{\mainc}, \{\tx_i\}) \rightarrow (\led_{\mainc}')$:] Takes as input the current mainchain ledger state $\led_{\mainc}$, and a set of pending transactions $\{\tx_i\}$, updates the ledger state based on the changes induced by these transactions and outputs updated ledger $\led_{\mainc}'$.
\end{description}

$\updateState$ extends the ledger with a new block containing new transactions based on the used consensus protocol. 

It should be noted that the notion above captures a base resource market. Resource markets in practice may have additional modules, with an additional set of transactions, and the market may offer additional non-service-related transactions beyond just $\tx_{\transfer}$. The model above can be extended to cover these additional semantics.

\systwo's functionalities can be abstracted as follows:
\begin{description}
    \item[$\setup(\led_{\mainc}^0, \modules) \rightarrow (\{\led_{\sidec, \module}^0\}_{\module \in \modules}, \led_{\mainc}')$:] Takes as input the mainchain genesis block $\led_{\mainc}^0$ and the set of the functional modules in the market $\modules$. It creates a sidechain for each module and outputs the initial state of these sidechains $\{\led_{\sidec, \module}^0\}_{\module \in \modules}$. $\led_{\sidec, \module}^0$ is the genesis block containing the module/sidechain specific public parameters $\param_\module$ including traffic classification, summary rules, and sidechain consensus configuration. $\setup$ also creates the mainchain state variables that will be synced by the sidechains, thus producing an updated mainchain ledger state $\led_{\mainc}'$.

    \item[$\elect(\led_{\mainc}) \rightarrow \left(\{\{\com_{i,\module}\}, \{\leader_{i,\module}\}\}_{\module \in \modules}\right)$:] Takes the current state of the mainchain $\led_{\mainc}$  as input, and outputs a set of committees $\{\com_{i,\module}\}$ and their leaders $\{\leader_{i, \module}\}$ for all module sidechains., where $i \in \{1, \dots, n\}$ and $n$ is the size of the committee set. The first committee $\com_1$ in every $\{\com_i\}$ is the primary committee while the rest are backup committees.

    \item[$\createSyncTransaction(\aux) \rightarrow (\tx_{\sync, \module})$:] Takes as input information $\aux$, and outputs a sync-transaction $\tx_{\sync, \module}$ for a particular module sidechain.

    \item[$\verifySyncTransaction(\led_{\sidec, \module},\tx_{\sync, \module}) \rightarrow (0/1)$:] Takes as input a module's sidechain ledger state $\led_{\sidec, \module}$ and its corresponding sync-transaction $\tx_{\sync, \module}$, and outputs 1 if $\tx_{\sync, \module}$ is valid based on its syntax/semantics, and 0 otherwise.

    \item[$\verifyBlock(\led_{\sidec, \module}, \block_{\btype}) \rightarrow (0/1)$:] Takes as input a module's sidechain ledger state $\led_{\sidec, \module}$, a new block $\block$ of type $\btype = \meta$ or $\btype = \summary$. It outputs 1 if $\block$ is valid based on the syntax/semantics of the particular block type and the module, and 0 otherwise.

    \item[$\updateState(\led_{\sidec, \module}, \aux, \btype) \rightarrow (\led_{\sidec, \module}')$:] Takes as input a module's sidechain ledger state $\led_{\sidec, \module}$, and pending transactions $\aux = \{\tx_i\}$ if block type $\btype = \meta$, or $\aux = \bot$ if $\btype = \summary$ since the inputs for a summary-block are meta-blocks from $\led_{\sidec, \module}$. It reflects the changes induced by $\aux$ and outputs a new module sidechain ledger state $\led_{\sidec, \module}'$.

    \item[$\prune(\led_{\sidec, \module}) \rightarrow (\led_{\sidec, \module}')$:] Takes as input a module's sidechain ledger state $\led_{\sidec, \sidechain}$, and produces an updated state $\led_{\sidec, \sidechain}'$ in which all stale meta-blocks are discarded.
\end{description}

\noindent\textbf{Security model.} We aim to build a secure solution that preserves correctness and security of the underlying resource market. 
Starting with a secure market operating a secure blockchain (that satisfies liveness and safety~\cite{garay2015bitcoin}), deploying \systwo preserves the invariant that the market remains secure and processes its traffic based on its original logic. Besides ledger security, a secure market addresses all service-related threats~\cite{abc}, including service theft (a client does not pay for a received service), service slacking (a server collects payments without offering services), service corruption (a server delivers a corrupt service), and DoS attacks (avoid publishing transactions issued by particular parties).

\systwo aims to optimize the performance of a given resource market by boosting throughput and reducing confirmation delays as well as  blockchain size. It does that by splitting the processing of market service-related workload among the module sidechains. These sidechains process their workload based on the original logic of the market. Thus, in our security analysis, we show that \systwo and all the techniques it introduces do not violate the security of the underlying resource market.\medskip

\noindent\textbf{Adversary model.} We assume the underlying resource market and its blockchain to be secure. For the sidechains, we account for three miner behaviors: honest, malicious who deviate arbitrarily, and honest-but-lazy who passively collaborate with the adversary, e.g., accepting blocks without validation. The adversary can corrupt miners but without going above the corruption threshold of the consensus protocol. As \systwo adopts a PBFT-based consensus, the threshold of malicious miners in a committee cannot exceed $f$, where the committee size is $3f+1$ and $2f+1$ votes are needed to reach agreement. Similar to prior work~\cite{Kokoris18,Kogias16,garay2015bitcoin,pass2017analysis}, we assume bounded message delivery, so any message will be delivered within $\Delta$ time delay. As we deal with public blockchains, the adversary can see all messages, reorder and delay them within $\Delta$ delay. The adversary is slowly adaptive~\cite{avarikioti2023divide}; it can corrupt miners only at the beginning of an epoch. Lastly, we deal with probabilistic polynomial time (PPT) adversaries who cannot break the security of the used cryptographic primitives with non-negligible probability.
 \section{\systwo Design}
\label{sec:design}

\systwo is a hybrid sharding-sidechain scaling framework that achieves the best of both worlds; supporting the high parallelism promised by sharding, while taking advantage of concise state change syncing, arbitrary data exchange and pruning offered by dependent-sidechains. \systwo does that by introducing a functionality-oriented way for parallelism. As shown in Figure~\ref{chainscale-diag}, it splits the resource market into logical functional modules, and assigns each module to a sidechain that handles this module's workload. By doing so, \systwo eliminates the costly cross-sidechain transactions, a problem that hampered sharding solutions that randomly distribute the workload among shards. 

In this section, we present the design of \systwo including system setup, architecture and operation, handling interruptions, and its security. To simplify the discussion, we do that assuming a base market functionality covering the basic modules discussed in Section~\ref{sec:background}. Nonetheless, the same logic applies if a market has more modules, or if it supports multiple services each with its own modules.

\begin{figure}[t!]
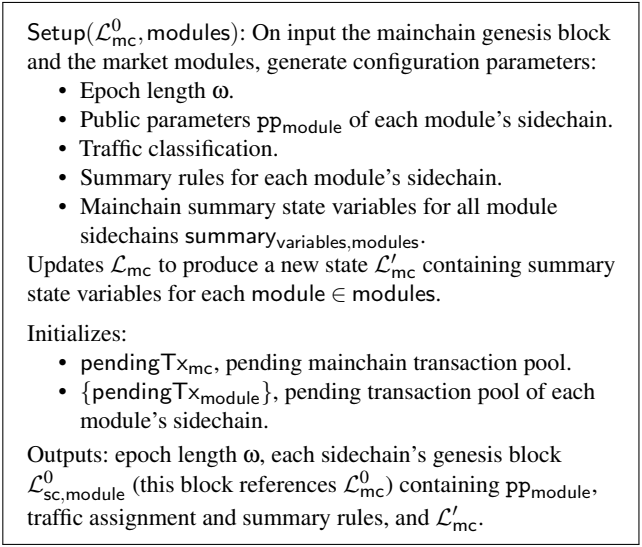

\raggedright
\begin{framed}
 \vspace{-3pt}
\small{
$\setup(\led_{\mainc}^0, \modules)$: On input the mainchain genesis block and the market modules, generate configuration parameters:
 \begin{itemize}[noitemsep,topsep=0pt]
     \item Epoch length $\omega$.   
     \item Public parameters $\param_{\module}$ of each module's sidechain.
     \item Traffic classification.
     \item Summary rules for each module's sidechain.
     \item Mainchain summary state variables for all module sidechains $\mathsf{summary_{variables, \modules}}$.
 \end{itemize}
 
 Updates $\led_{\mainc}$ to produce a new state $\led_{\mainc}'$ containing summary state variables for all modules.
 
 \vspace{4pt}
 Initializes:
 \begin{itemize}[noitemsep,topsep=0pt]     
     \item $\mathsf{pendingTx_{\mainc}}$, pending mainchain transaction pool.
     \item  $\mathsf{\{pendingTx_{\module}\}}_{\module \in \modules}$, a set of pending transaction pools such that each module's sidechain has its own pool.
 \end{itemize}

\vspace{2pt}
 Outputs: epoch length $\omega$, each sidechain's genesis block $\led_{\sidec, \module}^0$ (this block references $\led_{\mainc}^0$) containing $\param_{\module}$, traffic assignment and summary rules, and $\led_{\mainc}'$.
 }
 \vspace{-4pt}
\end{framed}
\caption{System setup.} 
\label{fig:syssetup}
\end{figure}

\subsection{Setup Phase} 
\systwo is a service-type agnostic scalability framework; it can be used for any resource market regardless of the service type it offers. At the same time, \systwo's generic architecture can be customized based on the market's modules and traffic patterns. In the setup phase (Figure~\ref{fig:syssetup}), system designers determine the market's modules and traffic split among these modules. Each module will have a sidechain that will process the transactions related to this module. System designers also specify how these transactions should be summarized and define the mainchain state variables that will be synced based on these summaries. All these factors, in addition to the epoch length, round duration, and block sizes, are set in each sidechain public parameters $\param_{\sidechain}$.\medskip 

\noindent\textbf{Module and traffic division.} Parallelism by workload sharing is effective so long as each module can commit its assigned workload on its own — i.e., no transaction writes state owned by another module. Thus, in our functionality-oriented approach, a module represents a well-defined part of the market that achieves a particular purpose. For our abstract model, these include market matching, service-payment exchange, and dispute resolution. All service-related transactions under each module are directed to the sidechain assigned to that module. For example, $\tx_\ask$ and $\tx_\offer$ would pertain to the market matching module, $\tx_\serviceProof$ and $\tx_\servicePayment$ belong to the service-payment exchange module, and $\tx_\dispute$ goes into the dispute module. 

The level of module granularity highly impacts system performance. For example, splitting the market matching module into two modules, one that handles $\tx_\ask$ and another that handles $\tx_\offer$, would lead to high volumes of (costly) cross-sidechain transactions. While merging, for example, the market matching and the dispute resolution modules into one would lead to an overloaded sidechain.\medskip

\noindent\textbf{Traffic classification.} We use a simple static approach, in which a transaction header contains field indicating its destination chain. This field includes a chain indicator, stating whether a transaction must be processed by the mainchain or one of the sidechains, and if it is the latter, there will be a module indicator, stating the destination module's sidechain. For example, with one byte prefix, first 2 bits of 00 means a mainchain transaction, while 11 means a sidechain transaction with the next 6 bits indicate which sidechain, e.g., 000001 is for the market matching sidechain, and so on. Figure~\ref{fig:load-div} shows an annotation example based on our abstract market model.\medskip

\begin{figure}[t]
    \centering
    \includegraphics[height= 0.7 in, width = 1.0\columnwidth]{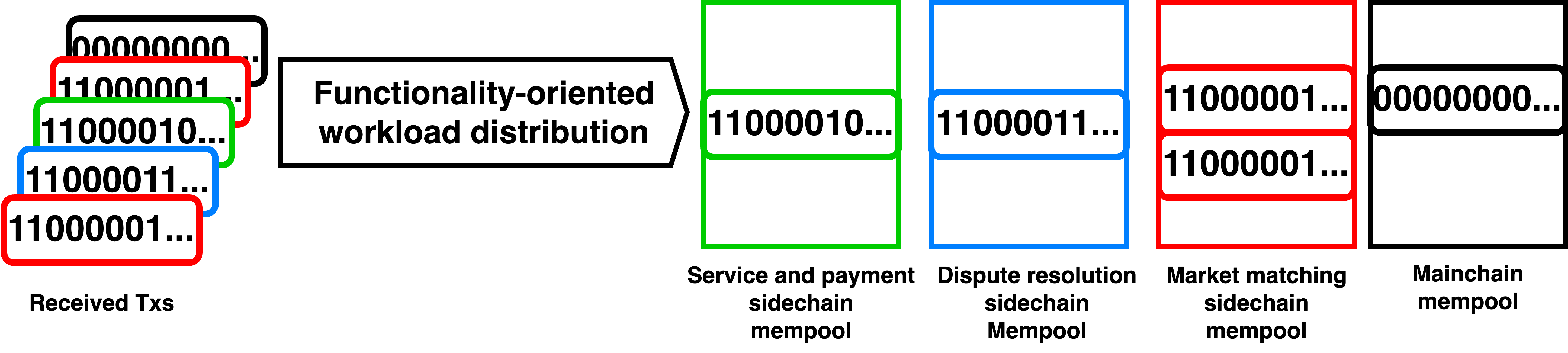}
    \vspace{-15pt}
    \caption{Annotated functionality-oriented load splitting.}
    \label{fig:load-div}
\end{figure}

\noindent\textbf{Transaction summary rules.} These are similar to those of \sysone. For example, $\tx_\servicePayment$ and $\tx_\serviceProof$ are summarized by counting the number of proofs a server submitted in an epoch and its total payment, $\tx_\dispute$ are summarized by listing the outcome of a dispute and its proof, while $\tx_\ask$ and $\tx_\offer$ are summarized by the finalized agreements they produce encapsulated in $\tx_{\agreement}$. Different from \sysone, in \systwo summaries are distributed among the modules' sidechains based on the traffic they handle. Figure~\ref{fig:summary-rules} presents the summary rules for the abstract resource market, which can be extended and adapted to match any additional modules defined during system setup.

\subsection{Sidechain Management}
The mainchain and sidechains in \systwo work in parallel. During each epoch, a committee from the mainchain miners is elected for each sidechain. Each committee runs a PBFT-based consensus protocol. Every round within an epoch, the committee agrees on their shared state, then produces a new meta-block recording all processed transactions in that round. Then, in the last round of an epoch, the committee produces a summary-block summarizing the state changes induced by the meta-blocks produced within the epoch. 

To reflect the state changes induced by sidechain transaction processing, each sidechain committee issues a sync-transaction (immediately after producing a summary-block) to sync the mainchain based on the summaries. Thanks to module splitting, sync-transactions of the sidechains in \systwo can be processed independently. Also, since the mainchain stores the system state, any module sidechain can retrieve the needed state changes from the mainchain. Once a sync-transaction is confirmed on the mainchain, the corresponding epoch meta-blocks on the corresponding sidechain are discarded. 

\begin{figure}[t!]
\raggedright
\begin{framed}
\vspace{-4pt}
\small{
Input: Epoch meta-blocks $\block^{\meta,1}_{\module}, \dots, \block^{\meta,n}_{\module}$ for $\module \in \{\service, \dispute,\marketmatch\}$.

\vspace{2pt}
Initialize a summary structure $\summ_\module$ for each module.

\vspace{2pt}
\forloop{$i \in \{1, \dots, n\}$ and every $\tx \in \block^{\meta,i}_{\module}$}

\;\;\;\;\ifcond{ $\module = \service$}

\;\;\;\;\;\;\ifcond{ $\tx.\txtype = \tx_{\serviceProof}$ }

\;\;\;\;\;\;\; // $\cid$ is the service contract ID

\;\;\;\;\;\;\;\; ++ $\summ_\module[\tx.\cid].\countt$

\;\;\;\;\;\;\elsecond \ifcond{ $\tx.\txtype = \tx_{\servicePayment}$ }

\;\;\;\;\;\;\;\; $\summ_\module[\tx.\cid].\payment$ += $\tx.\amount$

\;\;\;\;\elsecond \ifcond{ $\module = \dispute$ }

\;\;\;\;\;\; $\summ_\module[\tx.\cid] = (\tx\mathsf{.proof}, \tx\mathsf{.outcome})$ 

\;\;\;\;\elsecond \ifcond{ $\module = \marketmatch  $}

\;\;\;\;\;\; // $\tx$ under this module refers to $\tx_{\agreement}$ which 

\;\;\;\;\;\; // references $\tx_{\ask}$ and $\tx_{\offer}$, so the summary covers 

\;\;\;\;\;\;// these as well ($\mathsf{s}$ is the server and $\mathsf{cl}$ is the client)

\;\;\;\;\;\; $\summ_\module[\tx.\cid]$ = 
$(\tx\mathsf{.s}, \tx\mathsf{.cl},\tx\mathsf{.terms})$

Output $\summ_\module$ for all modules.
}
\vspace{-4pt}
\end{framed}
\caption{Summary rules (assuming one server per service contract, if the market assigns multiple servers per contract then indexing should indicate the server ID).} 
\label{fig:summary-rules}
\end{figure}

Like a regular blockchain, all miners have a copy of the mainchain. In \systwo, they also have copies of all sidechains. This is done to speed up committee rotation. That is, when a committee for a sidechain is elected, each committee member must have the latest state of this sidechain ledger so they can work on extending it. By requiring miners to have copies of the sidechains, committees will avoid the excessive delays of retrieving the sidechain ledger from the prior epoch committee. Furthermore, this also allows miners in the system to read state owned by the different system modules as a local lookup.

Note that correctness of layer 2 (i.e., sidechain) transaction processing is established by the security of consensus; the sidechain committee processes its workload using the original logic of the resource market and only agrees on valid state changes. No need for additional costly mechanisms such as ZK proofs or contestation periods as in rollups.\medskip

\noindent\textbf{Weighted miner assignment.} \sysone and sharding solutions employ a random committee election approach, which is also supported in \systwo. However, \systwo recognizes that some modules are more critical than others in their impact on market operation, so interruptions in these modules are more devastating than others. For example, the dispute solving module is critical as it decides sanctions for misbehaving parties and affects their participation in providing/receiving service. Given that miners would have different vested interest in the system, those with large involvement are incentivized to behave honestly. By assigning them to critical modules, we can reduce the chances of interruptions and even have fast recovery if interruptions happen. 

To this end, we devise a \emph{weighted miner committee election}, facilitating cryptographic sortition to run this election locally (thus reducing overhead). It classifies miners into classes based on their participation (represented as a score) in the system. Then, it allocates a plurality of the critical sidechain committee memberships to miners positive participation, while membership of less critical committees will be more balanced, i.e., have same number of miners from all classes on expectation. 

A simple example is shown in Figure~\ref{miner-div-new} assuming two sidechains $A$ and $B$ ($A$ is more critical than $B$) each requiring a committee of 6 miners. Here, we have $N=2n$ miners split into two classes where class 1 contains $n$ miners with the highest scores ($m_{i,1}$ and $m_{i, 2}$ refer to miners in class 1 and class 2, respectively). The more critical chain $A$ picks 4 miners from class 1 and the rest from class 2, while chain $B$ uses a balanced selection, picking 3 miners from each class.\medskip

\emph{The weighted miner committee election algorithm.} We capture positive participation via scores assigned to miners. These scores account for important parameters reflecting vested interest in the system. In detail, each miner calculates a score $s$ based on its mining power $P_m$, previous participation in sidechain consensus $S_c$, and the disputes it was involved in $D$. This is computed as $s = \alpha P_m + \beta S_c - \gamma D$, where $\alpha + \beta + \gamma = 1$ with values calibrated by system designers. The mining power $P_m$ can be computed by tracking how many blocks a miner mined (in proof-of-work), or its stake amount (in proof-of-stake), or amount of service it provides (in systems that rely on other resources), etc. For $S_c$, this is computed based on the number of agreements a miner has participated in for any sidechain (by observing signatures when running PBFT to agree on meta and summary blocks). While $D$ represent the number of disputes for which the miner is found guilty.\footnote{Recall that in many resource markets servers also play the role of miners, which is the assumption we have here.}

\begin{figure}[t!]
\begin{center}
\includegraphics[width = 0.85\columnwidth]{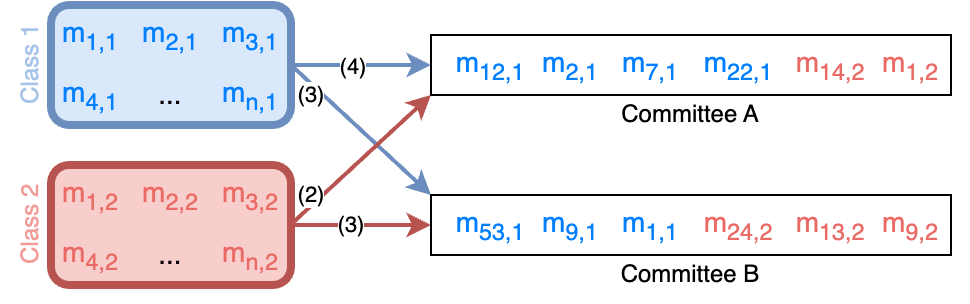}
\end{center}
\vspace{-8pt}
\caption{An example of weighted miner election, where a more critical module committee (A) is formed using 4 class 1 and 2 class 2 miners, and the other committee (B) has a more balanced selection, by selecting 3 miners from each class. }
\label{miner-div-new}
\end{figure}

Then, the miners in the system are ranked into $C$ classes based on this score, such that $C \geq |\modules|$. Say we have a total of $N$ miners, these miners are ranked based on their scores into $C$ classes of equal sizes. This is done by having the highest scoring $\mu = \lfloor\frac{N}{C}\rfloor$ miners be class 1, the following $\mu$ miners in class 2, and so on. Then, each module's sidechain will be randomly assigned a number of miners $n_{c, \module}$ from each class $c$, where $c \in \{1, \dots, C\}$, based on the module priority. Note that miners' scores are computed based on public values from the mainchain, allowing anyone to compute and verify the class of any miner. 

This construction also determines how the election behaves under Sybil attacks. Recall that miner identities are Sybil-resistant, so the mining power term 
$P_m$ cannot be duplicated: an adversary that spreads its holdings over several identities spreads $P_m$ among them as well. The participation term $S_c$ is constrained by the class ranking itself. Fresh identities all enter at zero score and populate the lowest class. Thus, adding identities enlarges the population competing for that fixed share and lowers the probability that any one of them is elected. Since committee seats are the only channel through which $S_c$ accrues, Sybil identities dilute their own access to it, and an adversary cannot accelerate its ascent into the high-participation classes. That ascent requires sustained participation across many epochs, which is precisely the behavior the score is meant to reward. The same effect raises the cost of abandoning an identity involved in disputes and re-entering with a clean one, as the replacement must re-earn both $P_m$ and $S_c$ from the lowest class.

\begin{algorithm}[t!]
\small{
\caption{Weighted miner election}
\label{alg:cap}
\begin{algorithmic}
\Procedure{Elect}{$seed_1, seed_2, s, S = [ s_i \text { for } i \in \{0, \cdots, N\}]$}
\State\textcolor{gray}{//Percentile of miner's score compared to all scores}
\State $\upsilon = 1-GetPercentile(s, S)$\medskip      

%\State
\State\textcolor{gray}{//The class of the miner from the percentile, recall that $C$ is} 
\State\textcolor{gray}{//the total number of classes}
\State $c = \lceil \upsilon \cdot C \rceil$\medskip        

%\State
\State\textcolor{gray}{//Total number of class $c$ miners needed for all sidechains.}
\State\textcolor{gray}{//For simplicity, we iterate over the modules}
\State\textcolor{gray}{//$\module \in \modules$ using index $i$.}
\State $n_{c, all} = \sum\limits_{i = 1}^{|\modules|} n_{c, \module_i}$\medskip

%\State
\State\textcolor{gray}{//Compute a miner-specific VRF output}
\State $rnd_1, \pi_1  \gets VRF_{sk}(seed_1||pk) $\medskip 

%\State
\State\textcolor{gray}{//Normalize VRF output and flip a biased coin to check if the}
\State\textcolor{gray}{//miner is elected for its class for any committee.}
\State\textcolor{gray}{//Recall that $\mu = \lfloor\frac{N}{C}\rfloor$, which is the size of any class.}
\State $elected = \frac{rnd_1}{2^{|rnd_1|}} < \frac{n_{c,all}}{\mu}$\medskip

\If{$elected$}
        \State \textcolor{gray}{//Compute another miner-specific VRF output.}
	\State $rnd_2, \pi_2 \gets VRF_{sk}(seed_2||pk)$
        %\State
        \State \textcolor{gray}{//Normalize VRF output and determine the module (i.e.,} 
        \State \textcolor{gray}{//the module sidechain assignment for the miner) based on}
        \State\textcolor{gray}{//the position in the range [0,1).}
        \If{$0\leq \frac{rnd_2}{2^{|rnd_2|}} < \frac{n_{c, \module_1}} {n_{c, all} } $}  
            \State {$ a \gets \module_1$ } 

        %\State 
        \ElsIf{$\frac{rnd_2}{2^{|rnd_2|}} < \frac{n_{c, \module_1} + n_{c, \module_2}}{n_{c, all} } $}
            \State {$ a \gets \module_2$ } 

        \State $\cdots$
        \Else 
        \State {$a\gets \module_{|modules|}$ }  \EndIf                            
\EndIf
\State \Return $s, rnd_1, \pi_1, rnd_2, \pi_2, a$
\EndProcedure
\end{algorithmic}
}
\end{algorithm}

%\mo{ Toy example: let $\mu = 10$ (each class contains 10 miners), and $n_{c,all} = 3$  (we need 3 miners from c-class), the probability that a miner $m$ is selected from c class would be would be $\frac{\binom{9}{2}}{{\binom {10} {3}}}$ --- $\binom{9}{2}$ is all the possible choices that contain miner $m$ [miner $m$ is fixed, so we are choosing 2 other miners, from a population of 9 to complete our choice of 3], and $\binom{10}{3}$ would be all the possible choices.  }

The election mechanism is depicted in  Algorithm~\ref{alg:cap}. Each miner computes its score $s$, and the score set of other miners $S$, and uses these to determine its class $c$ ($S$ can be cached locally and updated each epoch). A miner determines if it is elected by flipping a coin biased by the probability of its class being elected for any committee. This is done by, first, computing a miner-specific random value using a verifiable random function (VRF)~\cite{Gilad17} ($seed_1$ and $seed_2$ are public per-epoch fresh randomness computed in the system using, e.g., a random beacon). Then, the normalized VRF output is compared to the probability of being selected in any committee that requires class $c$ miners, which is $\frac{n_{c,all}}{\mu}$.\footnote{So, the miner-specific coin toss is $\frac{rnd_2}{2^{|rnd_2|}}$. The original expression for the bias, or target probability, is $\frac{\binom{\mu - 1}{n_{c, all} - 1 } }{\binom {\mu} {n_{c, all} } }$, which when simplified produces $\frac{n_{c,all}}{\mu}$. It captures the following: if you choose  $n_{c,all}$ miners uniformly at random from a population of $\mu$ $c$-class members (recall that $\mu$ is the total number of miners belonging to any class), what is the probability that a miner (of class $c$) is included among the $n_{c,all}$ miners selected.} If a miner is elected, then it determines the exact committee based on a second normalized VRF output. That is, each sidechain is assigned a sub-range within [0,1) based on the number of class $c$ miners it needs, and a miner finds the sub-range where its normalized VRF output belongs. At the end, the miner publishes its score and election result (VRF outputs, proofs, and assignment) so anyone can verify the election outcome.

In Appendix~\ref{apdx:class-comp}, we show how to configure the number of miners from each class to achieve a target committee failure probability (chosen based on module importance).

\textbf{Agreeing on a shared view.}
To ensure consistent block validation, members of each module's consensus committee must agree on a common view of the system. At the beginning of each consensus round for module $m$, every member of its consensus committee sends the leader a signed \emph{view report} containing, for each other chain ($\led_{\mainc}$ and $\led_{\sidec,j}$ for all $j \neq m$), the height of its latest finalized block and the corresponding PBFT commit certificate.\footnote{A commit certificate is the combination of committee member votes, i.e., their signatures, agreeing on the block.} The leader collects $2f+1$ valid reports and computes the component-wise minimum of the reported heights, identifying the longest prefix of each chain shared by all reporting members. These prefixes collectively form a \emph{snapshot}, which the leader includes in the proposed block's header. An honest committee member votes for the proposal only if the snapshot is (i) at least the snapshot of the proposed block's parent, (ii) a valid prefix of its local view, and (iii) all proposed transactions can be validated against it.

\subsection{Hierarchical Intra-module Scaling} 
In resource markets, it is expected that some modules would be more overloaded than others, making the functionality-oriented workload distribution imbalanced. For example, the service-payment exchange module, whose traffic scales linearly with the number of servers/clients, would be overloaded. To address this issue, \systwo introduces a \emph{hierarchical intra-module scaling} that employs the sharding-sidechain hybrid paradigm within the same module. This is done by spinning out sub-sidechains that process the module's traffic in parallel. To reduce on-chain traffic, the results of all sub-sidechains are summarized into one summary-block, and hence, only one sync-transaction is issued.

\begin{figure}[t]
\begin{center}
\includegraphics[height= 0.9 in, width = 0.9\columnwidth]{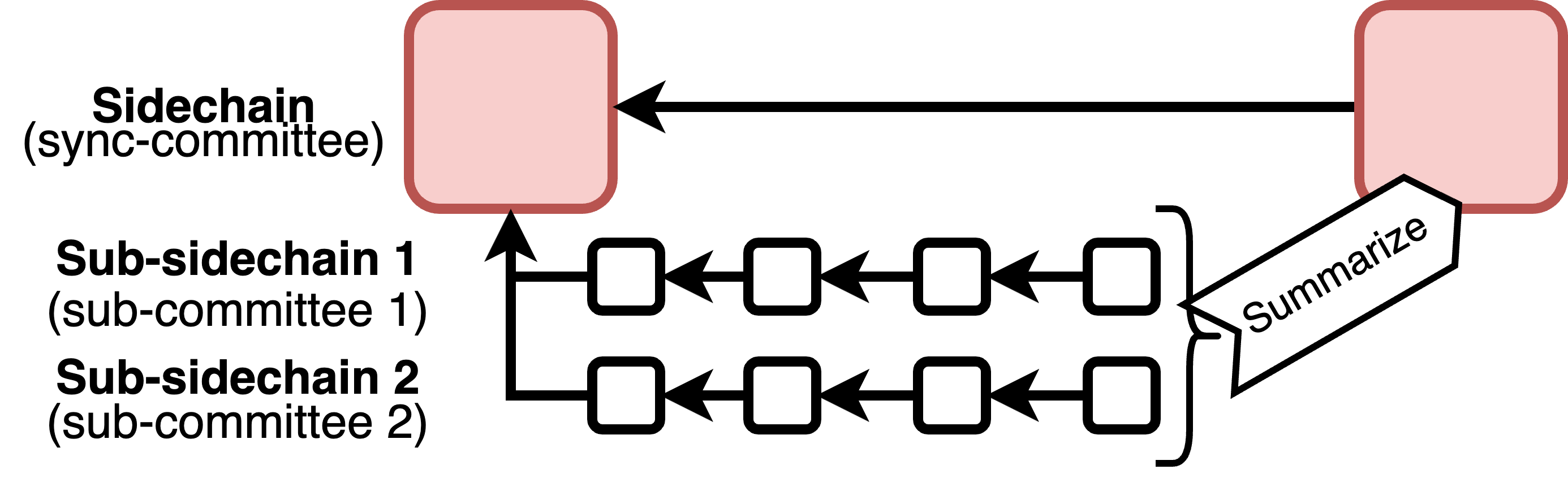}
\end{center}
\vspace{-10pt}
\caption{Scaling using sub-sidechains.}
\label{sub-sidechains}
\end{figure}

All sidechains with full meta-blocks, and whose mempools cannot be emptied in one epoch, are marked by their committees as heavy. The committee estimates the number of sub-sidechains needed to process the full traffic, and reports that in the sync-transaction. This triggers an increase in the number of committees to be elected for the next epoch to account for these sub-sidechains. If the current miner population cannot accommodate additional large-enough committees, then fewer sub-sidechains will be created.\footnote{The miner population size can be determined by tracking miner activity on the mainchain, or by having an identity sidechain where miners register when joining (similar to the identity shard in~\cite{Kokoris18}).} These sub-sidechains record only meta-blocks as shown in Figure~\ref{sub-sidechains}. A sync-committee will be elected from the members of the sub-committees (following the same distribution of these sub-committees) to produce one summary-block and issue a sync-transaction.

As an example, suppose miners are divided into 4 classes $(A,B,C,D)$, each of size 50 miners. The market has two heavy modules and a few non-heavy ones. Non-heavy modules need $(27A, 30B, 20C, 10D)$ miners for their committees, leaving $(23A, 35B, 30C, 40D)$ miners for the heavy modules. Assume the first module needs 5 sub-sidechains each with a committee of $(4A, 3B, 2C, 1D)$ miners, and the second module needs 3 sub-sidechains with a combined $(9A, 12B, 3C, 6D)$ miners. These cannot be satisfied based on the current miner population. Thus, we form committees as close as possible to modules' needs, having 4 sub-sidechains using $(16A, 12B, 8C, 4D)$ miners for the first module, and 2 sub-sidechains using $(6A, 8B, 2C, 4D)$ miners for the second.

When scaling a module, transactions are divided among the sub-sidechains randomly, in a way that guarantees that related transactions end up in the same sub-sidechain. In our abstract resource market, the service-payment exchange, and the dispute modules can be assigned based on their service contract identifier. For the market matching module, a standing ask or offer is posted randomly using its own identifier, and all transactions referencing it are routed to that same sub-sidechain. Agreements are assigned based on the service contract ID they define.

\subsection{Eliminating Cross-sidechain Transactions}

As mentioned before, system designers identify the modules of the resource market and assign a sidechain to each. This division follows the general functionality of resource markets, centered mainly around market matching, payment--service exchange, and dispute resolution. Each module's sidechain processes that module's transactions, while information originating elsewhere---e.g., checking that a client created an escrow---is read from the miner's local copy of the mainchain or of another sidechain (recall that in \systwo\ all miners maintain the mainchain and keep copies of the other sidechains to speed up committee switching). We now show that this functionality-oriented partition admits \emph{no cross-sidechain transactions}.

\begin{definition}[Cross-chain transaction]
Under an assignment of transactions to chains, a transaction $\tx$ is
\emph{cross-chain} if committing its effects atomically requires a decision from
more than one consensus instance.
\end{definition}

\begin{definition}[Properties of a system's partition scheme]
A system's partition scheme satisfies
\begin{itemize}
  \item \emph{Write locality} if the write set of every transaction $\tx$ with home
        module $m$ is contained in $\led_{\sidec,m}$; and
  \item \emph{Local read availability} if, for every $\tx$ with home module $m$,
        each value in the read set of $\tx$ is determined by state that every
        honest member of $\com_m$ holds locally and agrees upon at execution time.
\end{itemize}
\end{definition}

\begin{lemma}\label{lem:props}
The functional partition of \systwo\ satisfies write locality and local read availability.
\end{lemma}
\begin{proof}
\emph{Write locality.} Since in \systwo modules were carved along functional boundaries, no transaction type has a write set spanning two ledgers.

\emph{Local read availability.} During a round, miners in committee agree on their shared state of all $\led_{\sidec,\cdot}$ and of $\led_{\mainc}$. Each sidechain runs PBFT and thus has deterministic one-block finality, so no honest miner ever holds a shared state
that is later reverted, or that diverges from another honest miner's. Hence, for any value $v$ read by $\tx$, every honest member of $\com_m$ resolves $v$ identically from local state.
\end{proof}

\begin{theorem}
A system's partition scheme satisfying write locality and local read availability has no cross-chain transactions. In particular, by Lemma~\ref{lem:props}, \systwo\ has none.
\end{theorem}
\begin{proof}
Fix $\tx$ with home module $m$ and committee $\com_m$. By local read availability, executing $\tx$ is a deterministic function of $\tx$ and the shared state already held by honest member of $\com_m$; no external committee needs to certify any input,
so $\com_m$ reaches agreement on the resulting write set within a single PBFT instance. By write locality, that write set lies entirely in $\led_{\sidec,m}$, so the decision of $\com_m$ suffices to commit it, and atomicity is that of a single consensus decision. No second instance participates, so $\tx$ is not cross-chain.
\end{proof}

\subsection{Resilience to Interruptions} 
\systwo's sidechains have a mutual-dependency relation with the mainchain as they share the market workload. \systwo also creates a module criticality hierarchy, e.g., the dispute module issuing penalties against misbehaving miners can impact market matching and service-payment exchange modules. Moreover, there is the mainchain rollback issue, where abandoning recent blocks on its blockchain may lead to losing sync-transactions. Thus, to ensure system operation continuity and security, \systwo needs an \emph{autorecovery protocol} to handle three types of interruptions: rollbacks, intra-module sidechain interruptions, and inter-module issues.

\textbf{Handling rollbacks and intra-module recovery.}  \systwo uses the same mechanisms as \sysone~\cite{chainboost-paper}. Mass-syncing allows a committee to issue a sync-transaction for its epoch and all prior epochs impacted by a rollback. While view-change addresses misbehaving committee leaders, and backup committees handle misbehaving primary committees, thus ensuring intra-module recovery.

\begin{figure}[t!]
    \centering
    \includegraphics[height= 1.0 in, width = 0.8\columnwidth]{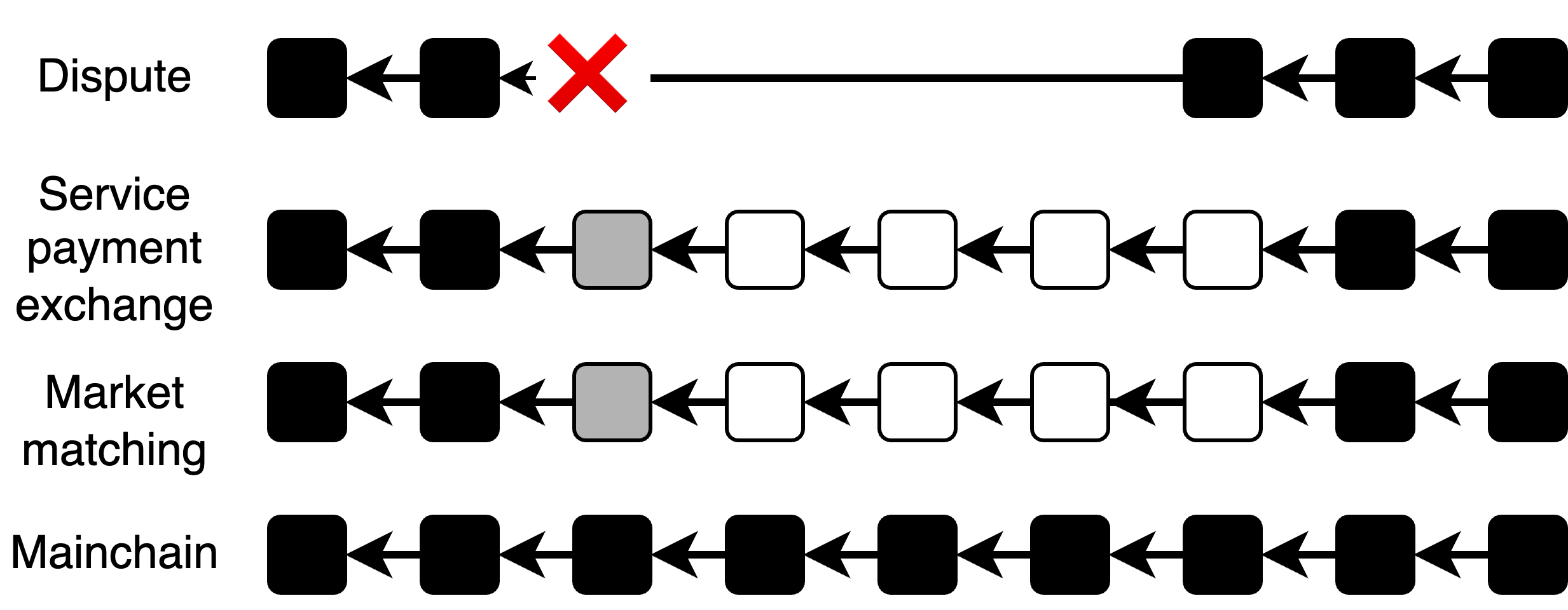}
    \vspace{-4pt}
    \caption{Autorecovery upon failure of the dispute sidechain (gray block signals interruption detection, and white blocks indicate empty blocks).}
    \label{fig:recov-epochs}
\end{figure}

\textbf{Inter-module issues.} We introduce a simple mechanism that allows a critical module to recover without impacting the liveness of other modules' sidechains. Say there is $\module_1$ whose behavior depends on information processed in a critical module $\module_2$. The miners in $\module_1$'s sidechain committee detect that $\module_2$ is interrupted if they do not see new blocks on $\module_2$'s sidechain after a timeout $\eta$. At that moment, $\module_1$'s committee mines empty meta-blocks until $\module_2$ recovers.  Figure~\ref{fig:recov-epochs} depicts this protocol; the dispute module is a critical module, e.g., the service-payment module must not process payments for a server that violates the service agreement. If the dispute module experiences an interruption, the two other modules stop transaction processing, and mine empty blocks until the dispute module recovers. The mainchain continues to operate normally but will have empty summaries from these sidechains during the interruption. 

\subsection{Security}
\label{sec-analysis}
In Appendix~\ref{apdx:sec-analysis}, we analyze the security of \systwo covering the failure probability of the autorecovery protocol, and showing that \systwo design preserves liveness and safety of the underlying resource market.

\section{Implementation and Setup}
\label{sec:implementation}

To evaluate the performance gains that \systwo can achieve, we implement a proof-of-concept prototype emulating a distributed file storage market, inspired by Filecoin~\cite{filecoin}, as a use case. The prototype defines modules and executes the \systwo protocol, using synthetic workloads corresponding to the patterns observed in Filecoin (discussed below). Execution progress is measured in logical blockchain rounds, each of which comprises the time to share the state and run consensus on the block.\footnote{Since the prototype does not emulate propagation delay, all committee members hold identical views, and so we do not implement view report exchange as it would not change our measurements. Its cost in practice is in-round latency that the round duration must accommodate.} Our code repository can be found at~\cite{chainscale-repo}.\medskip

\noindent\textbf{Sidechain implementation.} We extend the source code from~\cite{chainboost-paper} to have sidechains for the following modules:

\underline{\emph{Market matching}}: Here, a client asks for file storage service for a specific duration, file size, and price, via a contract-proposal transaction, and the server responds with a contract-response transaction that either approves or declines a client's ask. If a server accepts a proposal, it also issues a contract-deal transaction, which is the client's contract-proposal (originally signed by the client) signed also by the server; it sets the service terms and the client escrow used to pay for service. Otherwise, the two parties keep negotiating. In our implementation, the service contracts are assigned negotiation durations from a normal distribution. Whenever a contract expires, it enters into a negotiation phase that ends with the issuance of a contract-deal and reactivation of the contract.

\underline{\emph{Service-payment exchange}}: Servers prove storing files by submitting compact proofs-of-retrievability (PoR)~\cite{shacham2008compact} transactions. The total payments for valid submitted PoRs is dispensed to the server once the storage contract ends. 

\underline{\emph{Dispute resolution}}: When a miner observes a misbehavior (i.e., a server submits an invalid PoR transaction), it issues a dispute-transaction containing the proof of misbehavior. The miners in this module verify the proof and, if correct, publish the outcome on the module's sidechain. In our implementation, the misbehavior rate is a configurable parameter.

Each module has a single sidechain, except the service-payment exchange module that can have up to three sub-sidechains. All sidechains use same block size, round duration, and epoch length. For sidechain consensus, we use BLS Collective Signing (CoSi)-based PBFT inspired by~\cite{Kogias16} as implemented in~\cite{chainboost-repo} using the default committee election option (random cryptographic sortition). We evaluate the impact of our weighted sortition in a separate experiment.\medskip  

\noindent\textbf{Traffic generation.} We use the same traffic generation patterns and sizes from~\cite{chainboost-paper}, and augment them with dispute transactions (515 bytes) and contract-deals (716 bytes). We apply the traffic distribution observed in Filecoin using~\cite{site:filfox}; 98\% service-related transactions (destined to the sidechains) and 2\% currency transfers (destined to the mainchain). Transactions are created at the beginning of each mainchain round.\medskip

\noindent\textbf{Comparison to \sysone.} We extend the source code from~\cite{chainboost-paper} to support the new transactions introduced above. We use this extended implementation to run benchmark \sysone under the same workload as \systwo.\medskip

\noindent\textbf{Comparison to sharding.} We implement a proof-of-concept sharded resource market (that employs random workload distribution) inspired by Omniledger~\cite{Kokoris18} and Rapidchain~\cite{zamani2018}. Miners are assigned to shards randomly at the beginning of each epoch. Cross-shard transactions (those that have an input from another shard) are handled using inter-committee transaction forwarding as in Rapidchain~\cite{zamani2018}. That is, when a transaction has an input in another shard, a cross-shard transaction is forwarded to the destination shard who puts this transaction in its mempool to be processed.\medskip

\noindent\textbf{Evaluating autorecovery under different committee election mechanisms.} We write a C++ program to evaluate \systwo's autorecovery protocol under random and weighted committee elections. The program configures the percentages of lazy and malicious miners in a large miner population. Then, it creates one primary and several backup committees using both election approaches. Then, the program tracks the number of committee failures until the system recovers or fails (where a committee fails if the assigned number of adversarial nodes violates liveness). This allows quantifying the time needed to recover.\medskip

\noindent\textbf{Experiment setup.} We deploy our experiments on a computing cluster composed of 10 12-core, 130 GiB RAM, host-pinned VMs, connected with a 1 Gbps network link. Unless stated otherwise, our experiments run over a network of 8000 nodes, each serving 8 contracts. An experiment lasts for 61 mainchain rounds, with an epoch length $\omega = 10$ mainchain rounds corresponding to $\zeta = 30$ sidechain rounds, and a block size of 1 MB. A sidechain committee contains 500 miners. For the traffic pattern described before, 10\% of the active contracts generate disputes. We measure the time in mainchain blocks and convert to Filecoin's block duration: $\tau_{\mainc} = 30$ seconds. We collect experiment data using SQLite~\cite{website:sqlite}. In reporting our results, we measure the following metrics:

\underline{\emph{Throughput}}: The average number of transactions processed per mainchain round, including transactions appearing in a mainchain block $B_{\mainc}$ (or round) and those appearing in all sidechain blocks ${S}$ during the same mainchain round. We calculate throughput using the formula $T = \frac{\mathsf{Count}(B_{\mainc}.txs) + \sum_{B_{\sidec} \in S }\mathsf{Count}(B_{\sidec}.txs)} {\tau_{\mainc}}$

\underline{\emph{Confirmation time:}} The time (in seconds) a transaction takes to be confirmed from the moment it enters the queue. As we use PBFT, a transaction is confirmed once it is published in a meta-block.  This metric is given by the following formulas: $L_{\mainc} = (r_{c,\mainc} - r_{q, \mainc}) \tau_{\mainc} $, $L_{\sidec} = (r_{c,\sidec} - r_{q, \sidec}) \frac{\omega} {\zeta}\tau_{\mainc}$, where $r_{q, i} , r_{c,i} \text{ for } i \in \{\mainc, \sidec\}$ are the rounds where the transaction was queued and confirmed, respectively.

\underline{\emph{Storage footprint:}} The persistent blockchain size (mainchain blocks and sidechain summary-blocks) in megabytes. It is calculated as $S = \sum_{B_{\mainc} \in \mainc} |B_{\mainc}| + \sum_{B_{sum} \in \{B_{\sidec} \in \sidec | B_{\sidec}.type = summary\}}  |B_{sum}|$

\underline{\emph{Time to autorecover:}} The (worst case) average time needed to autorecover from interruption in a sidechain. This is obtained from the monte-carlo simulation.

After an experiment run ends, we continue processing traffic to empty the queues so we can measure confirmation times accurately. All other metrics are reported when the run ends.

\section{Performance Evaluation Results}
We test the impact of several parameters, including sidechain block size, ratio of sidechain to mainchain rounds in an epoch, and the number of sub-sidechains in heavy modules. We also compare \systwo to \sysone and sharding, and we study the impact of our weighted committee selection in reducing autorecovery overhead. To differentiate between configurations of \systwo, we notate them as $a$P$b$M$c$D-\systwo, where $a$ is the number of sub-sidechains in the service-payment exchange (P) module, $b$ is the number of sub-sidechains in the market matching (M) module, and $c$ is the number of sub-sidechains in the dispute (D) module.\medskip

\noindent\textbf{Scalability.}
We vary the number of service contracts within $\{32K, 64K, 128K, 512K\}$. During each epoch, a contract would result in a PoR transaction per mainchain round when active, a service payment when it expires, and a number of deal negotiation transactions before its re-activation. We compare \sysone to various configurations of \systwo. 

\begin{figure}[t]
\centering
\begin{subfigure}{.23\textwidth}
    \centering
    \includegraphics[width=\linewidth]{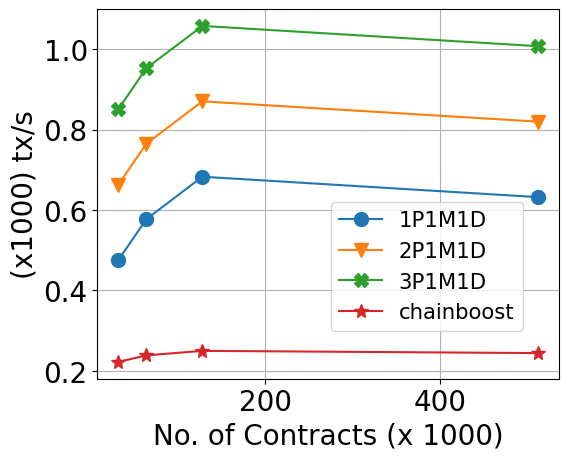}  
    \caption{Throughput}
    \label{subfig:scale-Throughput}
\end{subfigure}
\begin{subfigure}{.23\textwidth}
    \centering
    \includegraphics[width=\linewidth]{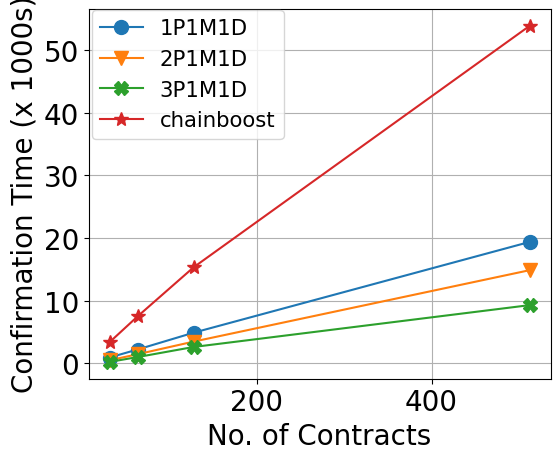}  
    \caption{Sidechain confirmation time}
    \label{subfig:scale-Latency}
\end{subfigure}
\begin{subfigure}{.23\textwidth}
    \centering
    \includegraphics[width=\linewidth]{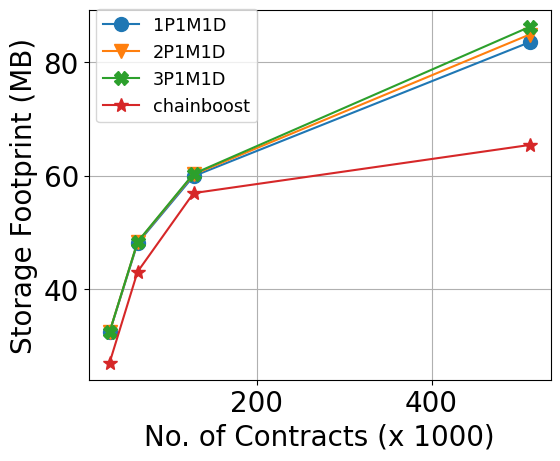}  
    \caption{Storage footprint}
    \label{subfig:scale-FootPrint}
\end{subfigure}
\vspace{-6pt}
\caption{Scalability results under various workloads.}
\vspace{-4pt}
\label{fig:scale}
\end{figure}

Throughput-wise, as shown in Figure~\ref{subfig:scale-Throughput}, \systwo significantly outperforms \sysone. At 32K contracts, 1P1M1D-\systwo and 3P1M1D-\systwo have 2x and 3.8x, respectively, the throughput of \sysone. We also notice that throughput peaks at 128K contracts, realizing a 4x increase. At 512K contracts, all versions of \systwo show a slight drop in throughput; this is due to the congestion in the system that exceeds the capacity of the sidechains.        

For confirmation time, \systwo also outperforms \sysone (Figure~\ref{subfig:scale-Latency}). At 32K contracts, this time is 3,380 s for \sysone while it is 944 s for 1P1M1D-\systwo---a 3.5x reduction in latency. The gap amplifies as workload increases; at 512K contracts, \sysone's confirmation time is 53,860~s, while it is 19,393 s for 1P1M1D-\systwo.

In terms of storage footprint, \systwo has a similar trend as \sysone until up to 128K contracts. For 512K contracts, we notice that \systwo needs around 20 MB of additional storage. This is because under congestion, \sysone (which has one sidechain) will process less workload (service contracts and disputes), leading to smaller summary-blocks. \systwo, with its multi-sidechain design, can deal with larger workloads, making its summary-blocks larger.\medskip

\noindent\textbf{Impact of block size.} We vary the sidechain block size within $\{0.5, 1, 1.5, 2\}$ MB, and compare \sysone to various configurations of \systwo in terms of throughput and confirmation time. 

\begin{figure}[t]
\centering
\begin{subfigure}{.23\textwidth}
    \centering
    \includegraphics[width=.95\linewidth]{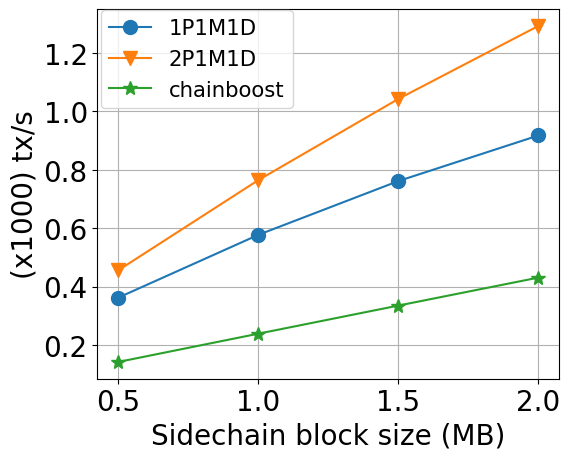}  
    \caption{Throughput}
    \label{subfig:blocksize-Throughput}
\end{subfigure}
\begin{subfigure}{.23\textwidth}
    \centering
    \includegraphics[width=.95\linewidth]{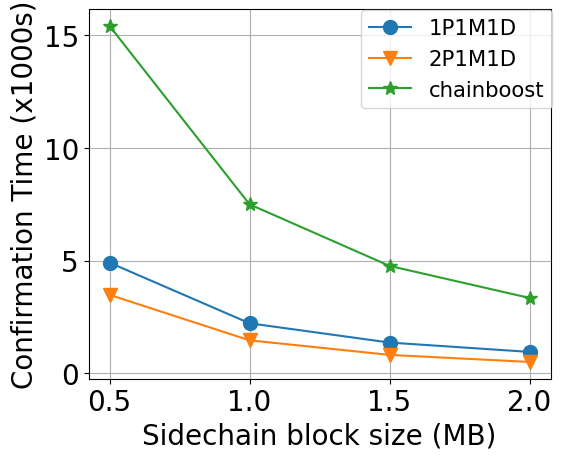}  
    \caption{Sidechain confirmation time}
    \label{subfig:blocksize-Latency}
\end{subfigure}
\caption{Impact of block size.}
\label{fig:blocksize}
\end{figure}

In Figure~\ref{subfig:blocksize-Throughput} we observe that \systwo outperforms \sysone in terms of throughput, even at smaller block sizes, where the difference is 2.5x under 1P1M1D-\systwo, and 3.22x under 2P1M1D-\systwo. As expected, \systwo's throughput improves at larger block sizes, achieving around 1293 tx/s under 2P1M1D an a block size of 2MB. For confirmation time, Figure~\ref{subfig:blocksize-Latency} shows that, at the smallest block size, 1P1M1D-\systwo and 2P1M1D-\systwo achieve 3.13x and 4.42x improvement, respectively, over \sysone (and these improvements reach 3.5x and 6.5 at the largest block size). This is again a result of the congestion that \sysone experiences as it has one sidechain.

Overall, \systwo with its multi-sidechain design and the hierarchical intra-module workload sharing for heavy modules, coupled with the right sidechain block size, achieves a high parallelism level. This makes \systwo a more viable solution for resource markets in practice than \sysone.\medskip

\noindent\textbf{Number of sidechain rounds per epoch.}
We vary this number within $\{30, 40, 60, 80, 100\}$. An epoch length in our setup is 10 mainchain rounds, so the sidechain round durations are $\{ 3, 3.75, 5, 7.5, 10\}$ s. We report the maximum throughput that can be achieved and the confirmation time.

\begin{figure}[t]
\centering
\begin{subfigure}{.23\textwidth}
    \centering
    \includegraphics[width=.95\linewidth]{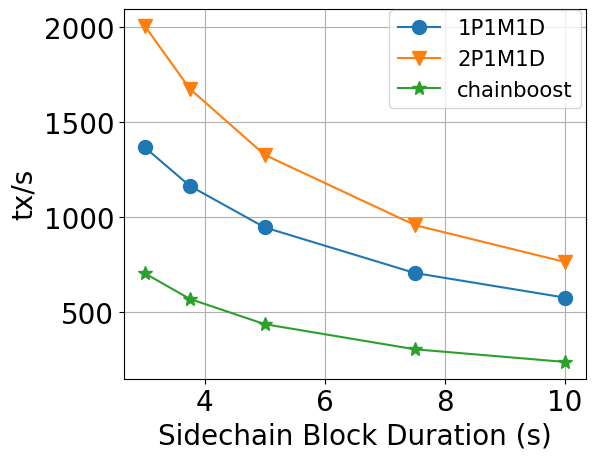}  
    \caption{Throughput}
    \label{subfig:blocktime-Throughput}
\end{subfigure}
\begin{subfigure}{.23\textwidth}
    \centering
    \includegraphics[width=.95\linewidth]{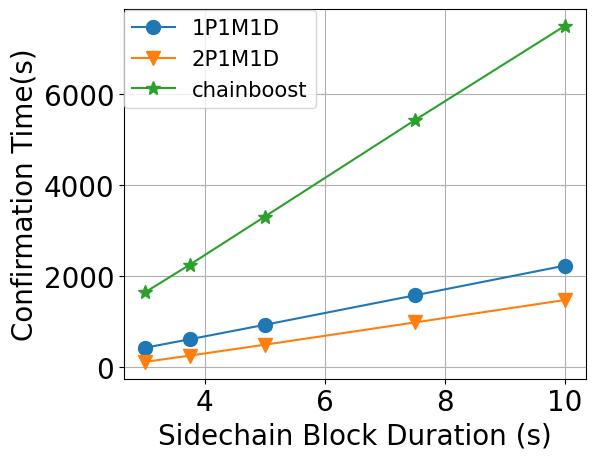}  
    \caption{Sidechain confirmation time}
    \label{subfig:blocktime-Latency}
\end{subfigure}
\caption{Impact of sidechain round duration.}
\label{fig:blocktime}
\end{figure}

Figure~\ref{subfig:blocktime-Throughput} shows that \systwo can achieve a throughput of 2000 tx/s when the sidechain round is 3 s (i.e., 10 sidechain rounds per mainchain round); a 2.66x improvement over \sysone. Throughput drops as the sidechain block duration increases because blocks now would contain the same number of transactions (block size is fixed) but are published at a slow rate. A similar trend is observed for confirmation time in Figure~\ref{subfig:blocktime-Latency}; for a 3 s sidechain round, \sysone has 3.93x and 5.11x the confirmation time achieved by 1P1M1D and 2P1M1D, respectively. Confirmation time increases for longer round durations. As such, the use of short rounds will promote the performance gains achieved by \systwo. However, system designers must account for the fact that a round duration must cover the time needed for consensus on the sidechain when configuring the system parameters.\medskip

\noindent\textbf{Comparison to sharding.} We compare \systwo to a sharding-based resource market (denoted as ShardedMarket). The sharded version has the same number of shards as the number of chains (both main and side chains) in \systwo's version of the market. We report throughput, confirmation time, storage footprint, the latency overhead added by cross-chain transactions, and their ratio (CTR) for a 1P1M1D-\systwo and ShardedMarket with four shards, under a workload of 128K contracts.

As shown in Table~\ref{tab:comp-shard}, \systwo outperforms sharding in terms of throughput and confirmation time. It achieves 2.5x the throughput of ShardedMarket. This is due to the sidechain design allowing for publishing blocks at a faster rate than the mainchain, thus processing traffic at a faster rate than sharding. \systwo achieves 3.5x improvement in latency. This is due to the fast traffic processing and the elimination of cross-chain transactions, which constitute 21.08\% of all transactions in ShardedMarket, and add 2966 seconds of latency overhead to the confirmation time. Storage-wise, since ShardedMarket does not prune blocks, the system's total storage grows to 244 MB, which is 3x the cost incurred by \systwo.\medskip

\noindent\textbf{Autorecovery overhead.} We evaluate the positive impact of \systwo's weighted miner election mechanism on system resilience. We configure a network of $10^6$ nodes, of which $(p_l, p_A) \in \{0.25, 0.3, 0.33\}^2$ are lazy and malicious miners, respectively. We organize the network into committees of size $S_c = 10^4$. We set the time needed for a backup committee to detect misbehavior and step in to be 5 min (Following from Appendix D.3 in~\cite{chainboost-paper}, this is the worst case when the system detects misbehavior at the end of an epoch, which spans 5 min in our experiments). We compare the default (random) and weighted committee elections under two miner classes C1 and C2. For C1, we vary its configuration to have 15\%, 30\%, and 45\% of the total adversarial nodes (C2 would contain the rest of the adversarial nodes---85\%, 70\%, and 55\%, respectively). For the weighted election, we vary the committee configuration to have 40\%, 50\%, and 60\% of its members from C1 and the rest from C2. We use the notation WXX-AYY, where XX is the percentage chosen from C1 and YY is the adversarial node percentage in C1. We conduct 10000 runs and report the average time needed to recover.  

\begin{table}[t!]
    \centering
    \caption{Comparing ShardedMarket vs. 1P1M1D-\systwo.}
    \label{tab:comp-shard}
  \resizebox{\columnwidth}{!}{  
    \begin{tabular}{ lcccccc}
        \toprule
       \multirow{2}{*}{\bf{Solution}}  &  Throughput & Confirmation & CT Latency & CTR & Storage \\
                 &  (tx/s) & Time (s) & Overhead (s) & \% & (MB) \\
       \midrule
        \bf{ShardedMarket} & 269.48 & 13419.92 &  2966.09 & 21.08 & 244\\
        \bf1P1M1D-\systwo & 672.38 & 3868.75 & 0 & 0 & 62 \\
       \bottomrule
    \end{tabular}
}
\end{table}

Our results can be found in Figure~\ref{fig:auto-recovery}. As expected, under both random and weighted election, time to recover (which includes the time needed for backup committees step in until a safe committee is found---if any) depends on the percentage of misbehaving nodes in the miner population. A larger percentage means a larger chance of having a misbehaving committee that violates liveness (i.e., the total number of non-voting nodes exceeds $f+1$ in the committee). However, weighted election (where the larger committee share is selected from a class with fewer adversarial nodes) performs better than random election, reducing the average autorecovery time by up to 3.31 min. This confirms the positive impact of our weighted approach in promoting security and reducing downtime of critical modules. 

Our experiments also show that a low-priority weighted sortition for less critical modules harms performance and may compromise security. In other words, using a weighted election where the majority of the committee miners are picked from classes with higher adversarial miners can slow down the autorecovery by up to 35 min, and even leads to failures. Thus, a more balanced approach (that is closer to random sortition) for less critical modules is more favorable.

\begin{figure}[t!]
%\captionsetup[subfigure]{font=small,labelfont=small}
\centering
\begin{subfigure}{.15\textwidth}
    \centering
    \includegraphics[width=.95\linewidth]{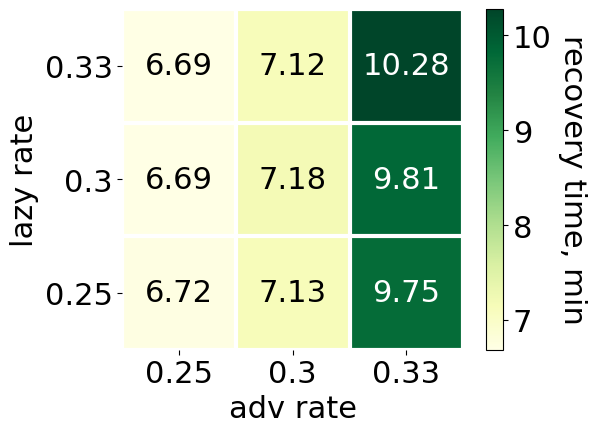}  
    \caption{Random}
    \label{subfig:auto-recov-random}
\end{subfigure}
\\
\begin{subfigure}{.15\textwidth}
    \centering
    \includegraphics[width=.95\linewidth]{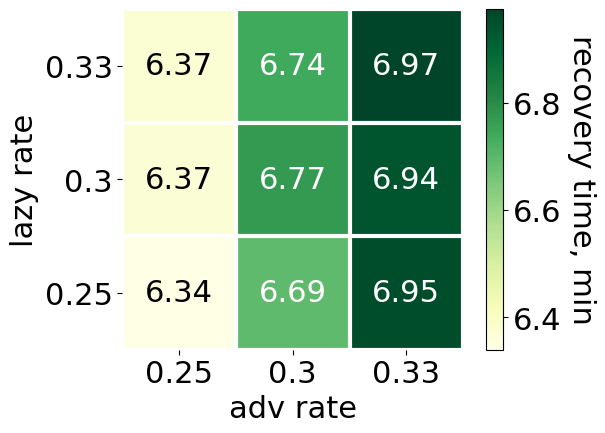}  
    \caption{W60\%-A15\%}
    \label{subfig:auto-recov-w-15}
\end{subfigure}
\begin{subfigure}{.15\textwidth}
    \centering
    \includegraphics[width=.95\linewidth]{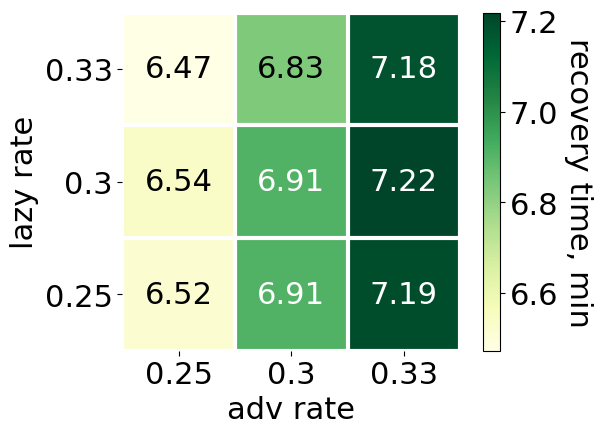}  
    \caption{W60\%-A30\%}
    \label{subfig:auto-recov-w-30}
\end{subfigure}
\begin{subfigure}{.15\textwidth}
    \centering
    \includegraphics[width=.95\linewidth]{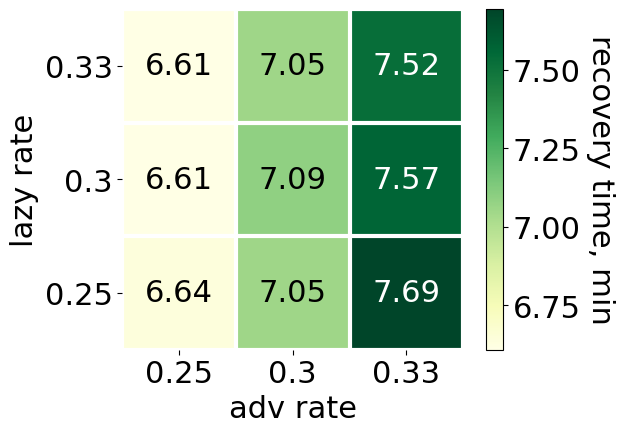}  
    \caption{W60\%-A45\% }
    \label{subfig:auto-recov-w-45}
\end{subfigure}

\begin{subfigure}{.15\textwidth}
    \centering
    \includegraphics[width=.95\linewidth]{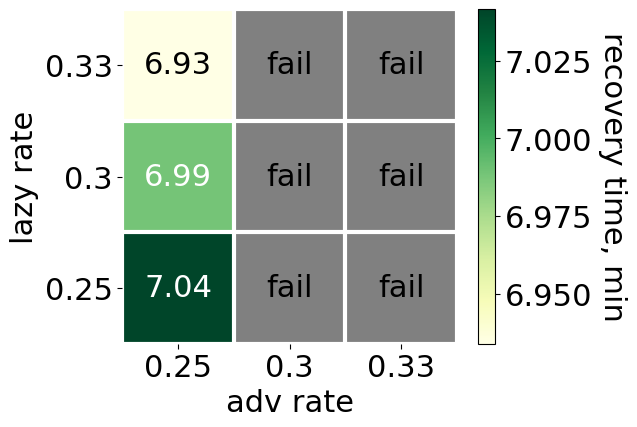}  
    \caption{W40\%-A15\% }
    \label{subfig:auto-recov-40-w-15}
\end{subfigure}
\begin{subfigure}{.15\textwidth}
    \centering
    \includegraphics[width=.95\linewidth]{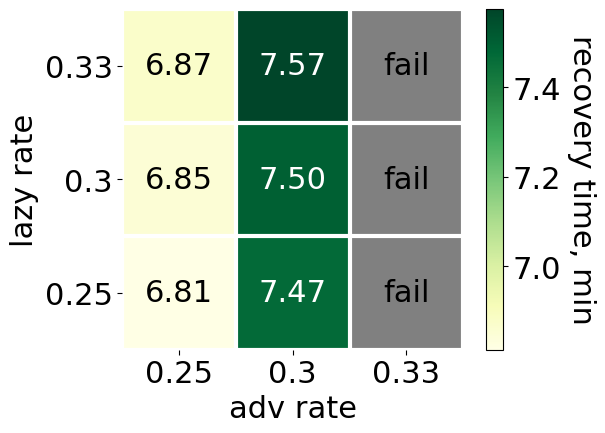}  
    \caption{W40\%-A30\%}
    \label{subfig:auto-recov-40-w-30}
\end{subfigure}
\begin{subfigure}{.15\textwidth}
    \centering
    \includegraphics[width=.95\linewidth]{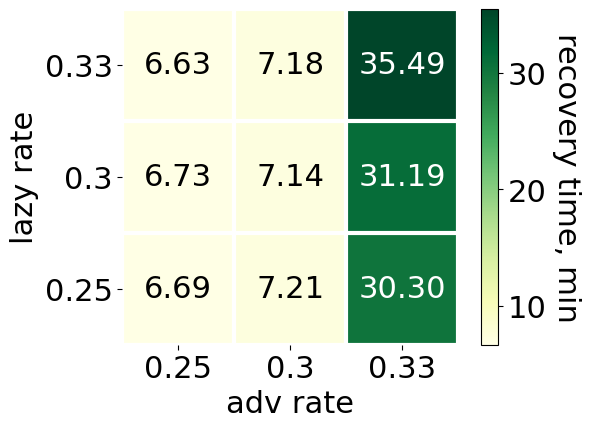}  
    \caption{W40\%-A45\% }
    \label{subfig:auto-recov-40-w-45}
\end{subfigure}

\begin{subfigure}{.15\textwidth}
    \centering
    \includegraphics[width=.95\linewidth]{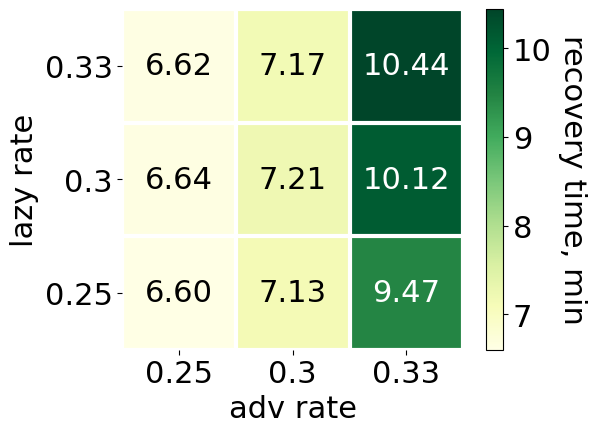}  
    \caption{W50\%-A15\%} 
    \label{subfig:auto-recov-50-w-15}
\end{subfigure}
\begin{subfigure}{.15\textwidth}
    \centering
    \includegraphics[width=.95\linewidth]{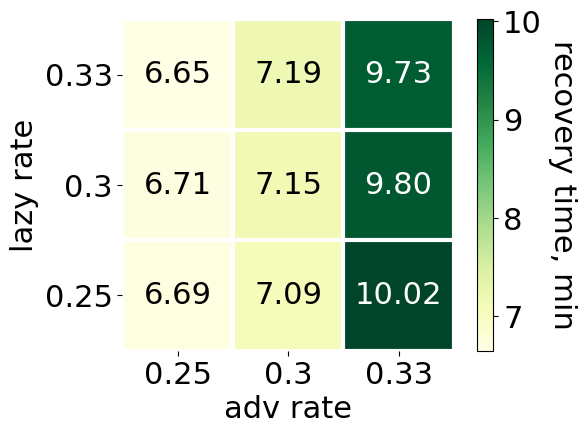}  
    \caption{W50\%-A30\%}
    \label{subfig:auto-recov-50-w-30}
\end{subfigure}
\begin{subfigure}{.15\textwidth}
    \centering
    \includegraphics[width=.95\linewidth]{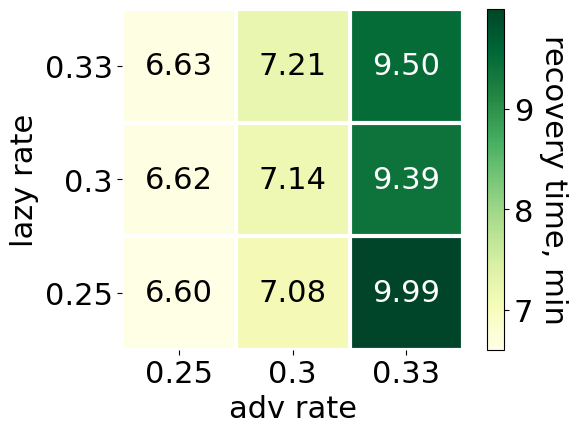}  
    \caption{W50\%-A45\%}
    \label{subfig:auto-recov-50-w-45}
\end{subfigure}
\caption{Average autorecovery time for random and 2-class weighted miner committee election under different adversarial distributions.}
\label{fig:auto-recovery}
\end{figure}

\section{Conclusion}
\label{conclusions}

We introduced \systwo, a secure layer-2 scaling solution for decentralized resource markets. \systwo encompasses a new hybrid sidechain-sharding architecture allowing for high parallelism of workload processing. It splits the resource market into functional modules, each of which is assigned its own sidechain. This is done without introducing any cross-chain transactions, which hampered earlier sharding solutions. We analyzed the
security of our system and conducted thorough performance
evaluations. Our results show that \systwo can achieve significant performance gains and improve threat resilience. 

\section*{Acknowledgment}
This work is supported by the National Science Foundation (NSF) under grant No. CNS-2226932.

\section*{AI Usage}
This paper was edited for grammar and editorial changes using Overleaf and Grammarly. The class composition and autorecovery analyses were verified using Claude. 

%\ack
%\artifacts
%We will open source our code repository that contains the full implementation as detailed in Section~\ref{sec:implementation}. For now, we offer an anonymized code repository that can be found at~\cite{chainscale-repo}.

\bibliographystyle{plainurl}
\bibliography{chainBib}

@inproceedings{shacham2008compact,
    title        = {Compact proofs of retrievability},
    author       = {Shacham, Hovav and Waters, Brent},
    year         = 2008,
    booktitle    = {ASIACRYPT},
    doi = {10.1007/978-3-540-89255-7_7}
}

@misc{zksyncDataAvailability,
	title = {Data availability in ZKsync},
	note = {\url{https://docs.zksync.io/zksync-protocol/rollup/data-availability}}
}

@inproceedings{avarikioti2023divide,
  title={Divide \& scale: Formalization and roadmap to robust sharding},
  author={Avarikioti, Zeta and Desjardins, Antoine and Kokoris-Kogias, Lefteris and Wattenhofer, Roger},
  booktitle={Structural Information and Communication Complexity},
  year={2023},
  doi = {10.1007/978-3-031-32733-9_10}
}

@inproceedings{nguyen2019optchain,
    title        = {Optchain: Optimal Transactions Placement for Scalable Blockchain Sharding},
    author       = {Nguyen, Lan N and Nguyen, Truc DT and Dinh, Thang N and Thai, My T},
    year         = 2019,
    booktitle    = {ICDCS},
    doi = {10.1109/ICDCS.2019.00059}
}

@inproceedings{wang2019mono,
    title        = {Monoxide: Scale out blockchains with asynchronous consensus zones},
    author       = {Wang, Jiaping and Wang, Hao},
    year         = 2019,
    booktitle    = {USENIX NSDI},
    note={\url{https://dl.acm.org/doi/10.5555/3323234.3323243}}
}

@inproceedings{zamani2018,
    title        = {Rapidchain: Scaling blockchain via full sharding},
    author       = {Zamani, Mahdi and Movahedi, Mahnush and Raykova, Mariana},
    year         = 2018,
    booktitle    = {CCS},
    doi = {10.1145/3243734.3243853}
}

@misc{op-verifier,
    title        = {Can I run a Verifier - Optimism Docs},
    note = {\url{https://help.optimism.io/hc/en-us/articles/4413155125403-Can-I-run-a-verifier}}
}

@misc{chainboost-repo,
	title        = {chainBoost code},
	note         = {\url{https://github.com/CSSL-UConn/chainboost-release}}
}

@misc{chainscale-repo,
	title        = {marketPrism code},
	note         = {\url{https://github.com/CSSL-UConn/marketprism-release}}
}

@misc{golem,
    title        = {Golem},
    note         = {\url{https://golem.network/}}
}

@misc{livepeer,
    title        = {Livepeer},
    note         = {\url{https://livepeer.com/}}
}

@misc{Optimism, 
howpublished={\url{https://optimism.io/}}, title={Optimism}
}

@misc{zksync, 
howpublished={\url{https://zksync.io/}}, title={zksync}
}

@misc{session,
    title        = {Session},
    note         = {\url{https://getsession.org/}}
}

@misc{filecoin,
    title        = {Filecoin},
    note         = {\url{https://filecoin.io/}}
}

@misc{celer2024zk,
    title        = {{T}he {P}antheon of {Z}ero {K}nowledge {P}roof {D}evelopment {F}rameworks},
    note = {\url{https://blog.celer.network/2023/08/04/the-pantheon-of-zero-knowledge-proof-development-frameworks/}}
}

@misc{arbitrumStateArbitrums,
    title        = {{T}he state of {A}rbitrum's progressive decentralization | {A}rbitrum {D}{A}{O} - {G}overnance docs --- docs.arbitrum.foundation},
    note = {\url{https://docs.arbitrum.foundation/state-of-progressive-decentralization}}
}

@misc{zknationZIP4Reduce, 
    title={[zip-4] reduce the execution delay from 21 hours to 3 hours},
    note={\url{https://forum.zknation.io/t/zip-4-reduce-the-execution-delay-from-21-hours-to-3-hours/373/3}}
}

@misc{site:filfox,
    title        = {Gas Statistics - Filfox},
    note         = {\url{https://filfox.info/en/stats/gas}}
}

@misc{optimismWithdrawalFlow,
	title = {Withdrawal flow --- docs.optimism.io},
	note = {\url{https://docs.optimism.io/stack/transactions/withdrawal-flow}}
}

@inproceedings{abc,
    title        = {ABC: A Cryptocurrency-Focused Threat Modeling Framework},
    author       = {Almashaqbeh, Ghada and Bishop, Allison and Cappos, Justin},
    year         = 2019,
    booktitle    = {INFOCOM Workshops},
    doi = {10.1109/INFCOMW.2019.8845101}
}

@article{Anjum17,
    title        = {Survey on peer-assisted content delivery networks},
    author       = {Anjum, Nasreen and Karamshuk, Dmytro and Shikh-Bahaei, Mohammad and Sastry, Nishanth},
    year         = 2017,
    journal      = {Computer Networks},
    doi = {10.1016/j.comnet.2017.02.008}
}

@article{Back14,
    title        = {Enabling blockchain innovations with pegged sidechains},
    author       = {Back, Adam and Corallo, Matt and Dashjr, Luke and Friedenbach, Mark and Maxwell, Gregory and Miller, Andrew and Poelstra, Andrew and Tim{\'o}n, Jorge and Wuille, Pieter},
    year         = 2014
}

@inproceedings{kalodner2018arbitrum,
    title        = {Arbitrum: Scalable, private smart contracts},
    author       = {Kalodner, Harry and Goldfeder, Steven and Chen, Xiaoqi and Weinberg, S Matthew and Felten, Edward W},
    year         = 2018,
    booktitle    = {USENIX Security Symposium},
    note = {\url{https://dl.acm.org/doi/10.5555/3277203.3277305}}
}

@inproceedings{chaliasos2024analyzing,
  title={Analyzing and benchmarking ZK-rollups},
  author={Chaliasos, Stefanos and Reif, Itamar and Torralba-Agell, Adri{\`a} and Ernstberger, Jens and Kattis, Assimakis and Livshits, Benjamin},
  booktitle={AFT},
  year={2024},
  doi = {10.4230/LIPIcs.AFT.2024.6}
}

@inproceedings{Danezis16,
    title        = {Centrally banked cryptocurrencies},
    author       = {Danezis, George and Meiklejohn, Sarah},
    year         = 2016,
    booktitle    = {NDSS},
    doi = {10.14722/ndss.2016.23187}
}

@inproceedings{Kokoris18,
    title        = {Omniledger: A secure, scale-out, decentralized ledger via sharding},
    author       = {Kokoris-Kogias, Eleftherios and Jovanovic, Philipp and Gasser, Linus and Gailly, Nicolas and Syta, Ewa and Ford, Bryan},
    year         = 2018,
    booktitle    = {IEEE S\&P},
    doi = {10.1109/SP.2018.000-5}
}

@inproceedings{Gilad17,
    title        = {Algorand: Scaling byzantine agreements for cryptocurrencies},
    author       = {Gilad, Yossi and Hemo, Rotem and Micali, Silvio and Vlachos, Georgios and Zeldovich, Nickolai},
    year         = 2017,
    booktitle    = {SOSP},
    doi = {10.1145/3132747.3132757}
}

@inproceedings{Luu16,
    title        = {A secure sharding protocol for open blockchains},
    author       = {Luu, Loi and Narayanan, Viswesh and Zheng, Chaodong and Baweja, Kunal and Gilbert, Seth and Saxena, Prateek},
    year         = 2016,
    booktitle    = {CCS},
    doi = {10.1145/2976749.2978389}
}

@inproceedings{Kogias16,
    title        = {Enhancing bitcoin security and performance with strong consistency via collective signing},
    author       = {Kogias, Eleftherios Kokoris and Jovanovic, Philipp and Gailly, Nicolas and Khoffi, Ismail and Gasser, Linus and Ford, Bryan},
    year         = 2016,
    booktitle    = {USENIX Security Symposium},
    note = {\url{https://dl.acm.org/doi/10.5555/3241094.3241117}}
}

@inproceedings{feldman2004free,
    title        = {Free-riding and whitewashing in peer-to-peer systems},
    author       = {Feldman, Michal and Papadimitriou, Christos and Chuang, John and Stoica, Ion},
    year         = 2004,
    booktitle    = {ACM SIGCOMM workshop on Practice and theory of incentives in networked systems},
    doi = {10.1145/1016527.1016539}
}

@inproceedings{Fisch19,
    title        = {Tight proofs of space and replication},
    author       = {Fisch, Ben},
    year         = 2019,
    booktitle    = {EUROCRYPT},
    doi = {10.1007/978-3-030-17656-3_12}
}

@inproceedings{gai2021cumulus,
    title        = {Cumulus: A Secure BFT-based Sidechain for Off-chain Scaling},
    author       = {Gai, Fangyu and Niu, Jianyu and Tabatabaee, Seyed Ali and Feng, Chen and Jalalzai, Mohammad},
    year         = 2021,
    booktitle    = {IWQoS},
    doi = {10.1109/IWQOS52092.2021.9521363}
}

@inproceedings{gao2022pshard,
    title        = {Pshard: a practical sharding protocol for enterprise blockchain},
    author       = {Gao, Jianbo and Zhang, Jiashuo and Li, Yue and Hao, Jiakun and Wang, Ke and Guan, Zhi and Chen, Zhong},
    year         = 2022,
    booktitle    = {International Conference on Blockchain Technology and Applications},
    doi = {10.1145/3581971.3581987}
}

@inproceedings{garay2015bitcoin,
    title        = {The bitcoin backbone protocol: Analysis and applications},
    author       = {Garay, Juan and Kiayias, Aggelos and Leonardos, Nikos},
    year         = 2015,
    booktitle    = {EUROCRYPT},
    doi = {10.1007/978-3-662-46803-6_10}
}

@inproceedings{Kiayias19,
    title        = {Proof-of-work sidechains},
    author       = {Kiayias, Aggelos and Zindros, Dionysis},
    year         = 2019,
    booktitle    = {Financial Cryptography and Data Security},
    doi = {10.1007/978-3-030-43725-1_3}
}

@inproceedings{Gavzi19,
    title        = {Proof-of-stake sidechains},
    author       = {Ga{\v{z}}i, Peter and Kiayias, Aggelos and Zindros, Dionysis},
    year         = 2019,
    booktitle    = {IEEE S\&P},
    doi = {10.1109/SP.2019.00040}
}

@inproceedings{Decker15,
	title        = {A fast and scalable payment network with bitcoin duplex micropayment channels},
	author       = {Decker, Christian and Wattenhofer, Roger},
	year         = 2015,
	booktitle    = {Symposium on Self-Stabilizing Systems},
    doi = {10.1007/978-3-319-21741-3_1}
}

@inproceedings{lao2024unpacking,
  title={Unpacking Long-Latency Transactions in Ethereum},
  author={Lao, Chon Kit and Zhou, Sophie and Zhang, Luyao and Zhang, Fan and Wang, Kanye Ye},
  booktitle={ Workshop on Decentralized Finance and Security},
  year={2024},
  doi = {10.1145/3689931.3694909}
}

@misc{eth-throughput,
	title        = {Ethereum blockchain explorer},
	note         = {\url{https://etherscan.io/chart/tx}},
    author = {}
}

@misc{website:sqlite,
    title        = {SQLite},
    note= {\url{https://www.sqlite.org/index.html}}
}

@misc{eip-4844,
    title        = {EIP-4844},
    note= {\url{https://eips.ethereum.org/EIPS/eip-4844}}
}

@inproceedings{hong2021pyramid,
    title        = {Pyramid: A layered sharding blockchain system},
    author       = {Hong, Zicong and Guo, Song and Li, Peng and Chen, Wuhui},
    year         = 2021,
    booktitle    = {IEEE INFOCOM},
    doi = {10.1109/INFOCOM42981.2021.9488747}
}

@misc{tomshardwareGoogleDrive,
    title        = {{G}oogle {D}rive users are reporting the loss of months of data --- tomshardware.com},
    note = {\url{https://www.tomshardware.com/software/cloud-storage/google-drive-users-are-reporting-the-loss-of-months-of-data}}
}

@inproceedings{koegl2023attacks,
    title        = {Attacks on rollups},
    author       = {Koegl, Adrian and Meghji, Zeeshan and Pellegrino, Donato and Gorzny, Jan and Derka, Martin},
    year         = 2023,
    booktitle    = {Workshop on Distributed Infrastructure for the Common Good},
    doi = {10.1145/3631310.3633493}
}

@misc{zksync-era-finality,
    title        = {Finality - ZkSync Era Docs},
    note = {\url{https://era.zksync.io/docs/reference/concepts/finality.html}}
}

@article{lee2021hierarchical,
    title        = {Hierarchical Multi-Blockchain System for Parallel Computation in Cryptocurrency Transfers and Smart Contracts},
    author       = {Lee, Nam-Yong},
    year         = 2021,
    journal      = {Applied Sciences},
    volume       = 11,
    number       = 21,
    doi = {10.3390/app112110173}
}

@inproceedings{li2022achieving,
    title        = {Achieving scalability and load balance across blockchain shards for state sharding},
    author       = {Li, Canlin and Huang, Huawei and Zhao, Yetong and Peng, Xiaowen and Yang, Ruijie and Zheng, Zibin and Guo, Song},
    year         = 2022,
    booktitle    = {SRDS},
    doi = {10.1109/SRDS55811.2022.00034}
}

@article{li2023security,
    title        = {On the Security of Optimistic Blockchain Mechanisms},
    author       = {Li, Jiasun},
    year         = 2023,
    journal      = {Available at SSRN 4499357},
    note = {\url{https://dx.doi.org/10.2139/ssrn.4499357}}
}

@inproceedings{liu2024pianist,
  title={Pianist: Scalable zkrollups via fully distributed zero-knowledge proofs},
  author={Liu, Tianyi and Xie, Tiancheng and Zhang, Jiaheng and Song, Dawn and Zhang, Yupeng},
  booktitle={IEEE S\&P},
  year={2024},
  doi = {10.1109/SP54263.2024.00035}
}

@article{mizrahi2020blockchain,
    title        = {Blockchain state sharding with space-aware representations},
    author       = {Mizrahi, Avi and Rottenstreich, Ori},
    year         = 2020,
    journal      = {IEEE Transactions on Network and Service Management},
    volume       = 18,
    number       = 2,
    doi = {10.1109/TNSM.2020.3031355}
}

@inproceedings{chainboost-paper,
  title={chainboost: A Secure Performance Booster for Blockchain-based Resource Markets},
  author={Motaqy, Zahra and Najd, Mohamed E and Almashaqbeh, Ghada},
  booktitle={IEEE EuroS\&P},
  year={2024},
  doi = {10.1109/EuroSP60621.2024.00022}
}

@inproceedings{Connor17,
    title        = {Enhancing Bitcoin Transactions with Covenants},
    author       = {O’Connor, Russell and Piekarska, Marta},
    year         = 2017,
    booktitle    = {Financial Cryptography and Data Security},
    doi = {10.1007/978-3-319-70278-0_12}
}

@misc{mediumOptimisticRollup,
    title        = {Optimistic Rollup is Not Secure Enough Than You Think},
    note = {\url{https://medium.com/onther-tech/optimistic-rollup-is-not-secure-enough-than-you-think-cb23e6e6f11c}}
}

@inproceedings{pass2017analysis,
    title        = {Analysis of the blockchain protocol in asynchronous networks},
    author       = {Pass, Rafael and Seeman, Lior and Shelat, Abhi},
    year         = 2017,
    booktitle    = {EUROCRYPT},
    doi = {10.1007/978-3-319-56614-6_22}
}

@article{tao2024throughput,
    title        = {Throughput-Scalable Shard Reorganization Tailored to Node Relations in Sharding Blockchain Networks},
    author       = {Tao, Liping and Lu, Yang and Fan, Yuqi and Shi, Lei and Tan, Chee Wei and Wei, Zhen},
    year         = 2024,
    journal      = {IEEE Transactions on Computational Social Systems},
    doi = {10.1109/TCSS.2024.3406769}
}

@inproceedings{yang2003ppay,
  title={PPay: Micropayments for Peer-to-peer Systems},
  author={Yang, Beverly and Garcia-Molina, Hector},
  booktitle={CCS},
  year={2003},
  doi = {10.1145/948109.948150}
}

@misc{cosmos,
  title={Cosmos},
  note = {\url{https://cosmos.network/}},
}

@misc{polkadot,
  title={Polkadot},
  note = {\url{https://polkadot.com/}},
}

@inproceedings{zheng2021meepo,
    title        = {Meepo: Sharded consortium blockchain},
    author       = {Zheng, Peilin and Xu, Quanqing and Zheng, Zibin and Zhou, Zhiyuan and Yan, Ying and Zhang, Hui},
    year         = 2021,
    booktitle    = {ICDE},
    doi = {10.1109/ICDE51399.2021.00165}
}

@article{liu2024dynashard,
  title={DynaShard: Secure and Adaptive Blockchain Sharding Protocol With Hybrid Consensus and Dynamic Shard Management},
  author={Liu, Ao and Chen, Jing and He, Kun and Du, Ruiying and Xu, Jiahua and Wu, Cong and Feng, Yebo and Li, Teng and Ma, Jianfeng},
  journal={IEEE Internet of Things Journal},
  year={2024},
  publisher={IEEE},
  doi = {10.1109/JIOT.2024.3490036}
}

@article{klenke2010stochastic,
  title={Stochastic ordering of classical discrete distributions},
  author={Klenke, Achim and Mattner, Lutz},
  journal={Advances in Applied probability},
  volume={42},
  number={2},
  pages={392--410},
  year={2010},
  publisher={Cambridge University Press}
}

@book{shaked2007stochastic,
  title={Stochastic orders},
  author={Shaked, Moshe and Shanthikumar, J George},
  year={2007},
  publisher={Springer}
}

@book{mitzenmacher2017probability,
  author    = {Mitzenmacher, Michael and Upfal, Eli},
  title     = {Probability and Computing: Randomization and Probabilistic Techniques
               in Algorithms and Data Analysis},
  edition   = {2},
  publisher = {Cambridge University Press},
  year      = {2017}
}

@article{dubhashi1996balls,
  title={Balls and bins: A study in negative dependence},
  author={Dubhashi, Devdatt P and Ranjan, Desh},
  journal={BRICS Report Series},
  volume={3},
  number={25},
  year={1996}
}

\appendix
\section{Class Composition Analysis} 
\label{apdx:class-comp}

In this section, we determine the class composition of a weighted-sortition committee such that its failure probability is lower than that of a committee elected using random sortition.  We start by defining the distributions governing committee memberships, then define the conditions under which a committee fails, and derive an epoch-independent upper bound on the failure probability using stochastic dominance and a Chernoff bound. Finally, we use this bound to determine the number of miners $n_{i, module}$ to be sampled from each class $i$, for a module $\module$, subject to a target failure probability $F_w < F$, where $F$ denotes the committee failure probability when elected using random sortition (under the same failure definition).\medskip

\noindent\textbf{Notation.} For a module $\module$, miner class $i$, we denote the number of miners from class $i$ as $n_i = n_{i,\module}$, and the miners needed for all modules as $n_{i,all}$. During each epoch $e$, a class $i \in \{1, \cdots,  C\}$ contains $\mathscr{M}_{i,e}$ misbehaving (lazy or malicious) and $\mathscr{W}_{i,e}$ well-behaved miners. We denote the total population in the system as N, and assume that each class has exactly $\mu=\lfloor N/C\rfloor$ miners, so that $\mathscr{M}_{i,e}+\mathscr{W}_{i,e}=\mu$.

\begin{lemma}[Committee memberships distributions]\label{lem:bernoulli}
Under the weighted sortition algorithm, the numbers of misbehaving and well-behaved committee members from class $i$ in epoch $e$, $B_{i,e}$ and $G_{i,e}$, respectively, and those across all classes, $B_e$ and $G_e$, respectively, satisfy the following and all $B_{i,e}$ and $G_{i,e}$ mutually independent:
    \[
    \begin{aligned}       
        &B_e = \sum_{i=1}^{C} B_{i,e}, \quad B_{i,e} \sim \mathrm{Bin}(\mathscr{M}_{i,e}, q_i) \\
        &G_e = \sum_{i=1}^{C} G_{i,e}, \quad G_{i,e} \sim \mathrm{Bin}(\mathscr{W}_{i,e}, q_i)
    \end{aligned}
    \]

\noindent where $q_i = \frac{n_i}{\mu}$.
\end{lemma}

\begin{proof}
    Consider a class-$i$ miner $j$, and define the following events:
    \begin{itemize}
        \item $E_{j,e}$: miner $j$ is elected to participate in any committee in epoch $e$. This event happens with probability $n_{i,all}/\mu$, uniformly across epochs.

        \item $A_{j,e, \module}$: miner $j$ is assigned to the module $\module$ for epoch $e$, conditional on being elected. This event happens with probability $n_i/n_{i,all}$, uniformly across epochs.
    \end{itemize}

   A miner joins the committee of $\module$ if both events occur. Recall that our weighted sortition relies on a VRF-based approach. Since the two selection stages use independent VRF outputs, we have:
    \[
        \begin{aligned}
            \Pr[E_{j,e}\cap A_{j,e,\module}]
            &=\Pr[E_{j,e}]
            \Pr[A_{j,e,\module}\mid E_{j,e}]\\
            &=\left(\frac{n_{i,\mathrm{all}}}{\mu}\right)
            \left(\frac{n_i}{n_{i,\mathrm{all}}}\right)\\
            &=\frac{n_i}{\mu}=q_i.
    \end{aligned}
\]

Furthermore, the VRF-based selection across miners is run independently, thus each class-$i$ miner performs an independent Bernoulli trial with success probability $q_i$. Therefore,

    \[
        G_{i,e} \sim \mathrm{Bin}(\mathscr{W}_{i,e}, q_i)\; , \; B_{i, e} \sim \mathrm{Bin}(\mathscr{M}_{i,e}, q_i) 
    \]
Since the class populations are disjoint and selection trials are independent, all $G_{i,e}$ and $ B_{i, e}$  are mutually independent. Summing over these classes establishes $B_e$ and $G_e$ in the lemma. 
\end{proof}

\begin{lemma}[Committee Failure]\label{lem:sc-fail}
Suppose that for every class $i$ and every admissible epoch $e$, we have $\mathscr{M}_{i,e} \le\overline{\mathscr{M}}_i$ and $\mathscr{W}_{i,e} \ge\underline{\mathscr{W}}_i$. Define the mutually independent random variables:
\[
	\overline B =\sum_{i=1}^{C}\overline B_i, \qquad \overline B_i \sim\mathrm{Bin}(\overline{\mathscr{M}}_i,q_i),
\]
\[
	\underline G =\sum_{i=1}^{C}\underline G_i, \qquad \underline G_i \sim\mathrm{Bin}(\underline{\mathscr{W}}_i,q_i).
\]
And let
\[
    g_{\min}=\EE[\underline G] =\sum_{i=1}^{C}\underline{\mathscr{W}}_i q_i.
\]
Then, for a committee with expected elected size $\cs$ and minimum required participating size $s_{\min}$, where $0<s_{\min}\le g_{\min}$, its failure probability in any admissible epoch satisfies:
\[
	\Pr[sc\mhyphen fail] \le \Pr[2\overline B\ge\underline G] +\exp\left(-\frac{g_{\min}}{2} \left(1-\frac{s_{\min}}{g_{\min}}\right)^2\right).
\]
\end{lemma}

\begin{proof}

A weighted sortition committee, with size $S_e = B_e + G_e$ fails for two reasons: either it fails to hold its Byzantine fault threshold ($B_e \geq \lceil\frac{S_e}{3} \rceil$), or it produces a participating committee (members who publish their election proofs) of size  $S_e^{\mathrm{part}}$ that is too small to provide a prescribed level of fault tolerance (e.g.,  tolerating at least $f_{\min}$ misbehaving miners). The committee-failure event is therefore contained in
\[
    \{sc\mhyphen fail,e\} \subseteq \{B_e \geq \lceil\frac{S_e}{3} \rceil\} \cup \{S_e^{\mathrm{part}} < s_{\min}\}.
\]

\noindent where $s_{min}$ is the minimum committee size that can guarantee a minimum level of fault tolerance.

Recall that a committee exceeds its Byzantine fault threshold if and only if:
\[
    \begin{aligned}
     B_e\ge \left\lceil\frac{S_e}{3}\right\rceil 
        &\iff 3B_e\ge S_e\\
        &\iff 2B_e\ge G_e.
\end{aligned}
\]

To bound the probability $\Pr[2B_e\ge G_e]$, we need to understand its ordering against $\Pr[2\overline{B}\ge \underline{G}]$, using stochastic dominance. A binomial with fixed success probability is stochastically increasing in its number of trials\cite[Theorem~1.1(a)]{klenke2010stochastic}. Thus, for $i\in\{1,\ldots,C\}$ we have $B_{i,e} \preceq_{\mathrm{st}} \overline{B}_i$ and  $\underline{G}_i \preceq_{\mathrm{st}} G_{i,e}$. The second relation implies $-G_{i,e} \preceq_{\mathrm{st}}-\underline G_i$.

By Lemma~\ref{lem:bernoulli}, the variables $B_{i,e}, G_{i,e}$ are independent. Applying the closure of
stochastic order under sums of independent random variables \cite[Theorem~1.A.3(b)] {shaked2007stochastic}, we obtain
\[
\sum_{i=1}^{C} \biggr( 2B_{i,e} -G_{i,e} \biggl) \;\preceq_{\mathrm{st}}\;
\sum_{i=1}^{C} 2\overline{B}_i -\underline{G}_i
\]

Equivalently, we have $2B_e -  G_e \;\preceq_{\mathrm{st}}\; 2\overline{B} - \underline{G}$, and therefore:
\[
    \Pr[2B_e \ge G_e] \le \Pr[2\overline{B} \ge \underline{G}]
\]

For bounding the participating committee size failure, we note that since all elected well-behaved miners publish their proofs, the actual participating committee size $(S_e^{\mathrm{part}} \leq S_e)$ satisfies: 
\[
    S_e^{\mathrm{part}} =B_e^{\mathrm{part}}+G_e \ge G_e.
\]

\noindent where $B_e^{\mathrm{part}}$ is the number of misbehaving miners that publish their proofs. Consequently, the probability that fewer than $s_{\min}$ miners publish valid election proofs is bounded by
\[
    \Pr[S_e^{\mathrm{part}}<s_{\min}] \le \Pr[G_e<s_{\min}].
\]

By stochastic dominance, $\underline G\preceq_{\mathrm{st}}G_e$,
and hence

\[
\begin{aligned}
    & \Pr[G_e\geq s_{\min}] \geq \Pr[\underline G \geq s_{\min}] \\ 
    & \iff 1 - \Pr[G_e\geq s_{\min}]  \leq 1 - \Pr[\underline G \geq s_{\min}] \\
    & \iff \Pr[G_e<s_{\min}] \le \Pr[\underline G<s_{\min}]
\end{aligned}
\]

By construction, $\underline G$ is a sum of independent Bernoulli random variables with expectation:
\[
\EE[\underline G] =\sum_{i=1}^{C}\underline{\mathscr{W}}_i q_i = g_{\min}.
\]

The multiplicative Chernoff lower-tail bound, \cite[Theorem~4.5]{mitzenmacher2017probability} gives
\[
    \Pr[\underline G <(1-\delta) g_{\min}]
    \le
    \exp\left(-\frac{g_{\min} \delta^2}{2}\right), \qquad 0 < \delta < 1.
\]

Setting $\delta=1-\frac{s_{\min}}{g_{min}}$, we obtain
\[   
    \begin{aligned}
    \Pr[G_e<s_{\min}] &\leq \Pr[\underline G < s_{\min}]\\
                       &\leq
    \exp\left(-\frac{g_{\min}}{2}\left(1-\frac{s_{\min}}{g_{min}}\right)^2\right)
    \end{aligned}
\]

To determine the committee failure, we apply a union bound to the Byzantine fault-threshold and minimum size failure events, and we have 
\[
    \begin{aligned}
    \Pr[sc\mhyphen fail, e] & \le \Pr[ \{2B_e\ge G_e\} \cup \{S^{\mathrm{part}}_e<s_{\min}\} ]\\
                            &\le \Pr[2B_e\ge G_e] +\Pr[G_e<s_{\min}]\\
                            &\le \Pr[2\overline B\ge\underline G]\\
                              &\quad+ \exp\left(-\frac{g_{\min}}{2}
                             \left(1-\frac{s_{\min}}{g_{\min}}\right)^2\right).
    \end{aligned}
\]

Since the class-population bounds, the probability of being elected to a committee, and the expected committee size do not depend on the epoch, the resulting bound holds uniformly over all admissible epochs.
\end{proof}

\noindent\textbf{Numerical example.} We consider a system with $C = 3$ classes whose misbehaving fractions are bounded by $\overline{p}_1 = 0.15$, $\overline{p}_2 = 0.25$ and $\overline{p}_3 = 0.35$, i.e., $\overline{\mathscr{M}}_i = \lfloor \overline{p}_i \mu \rfloor$ and $\underline{\mathscr{W}}_i = \mu - \overline{\mathscr{M}}_i$. A weighted-sortition committee draws $n_1 = 349$, $n_2 = 249$ and $n_3 = 149$ miners in expectation from the three classes, for an expected committee size of $\cs = 747$. We require a minimum committee size of $s_{\min} = 400$, which corresponds to tolerating at least $f_{\min} = 133$ Byzantine members. We compute the term $\Pr[2\overline{B} \ge \underline{G}]$ of Lemma~\ref{lem:sc-fail} exactly, and bound the minimum-size term with the Chernoff bound, using $g_\min = \sum_i \underline{\mathscr{W}}_i q_i = 580.25$.

For comparison, we evaluate random sortition under the same failure definition: a single population of $N = C\mu$ miners, 25\% of whom are misbehaving, from which each miner is elected independently with probability $\cs/N$, giving the same expected committee size of 747. Table~\ref{tab:num-example} reports both bounds for different class sizes.

Across all considered class sizes, weighted sortition yields a failure probability bound approximately four to five orders of magnitude smaller than that of random sortition, despite both schemes having the same expected committee size.

\begin{table}[h]
\centering
\caption{Failure probability bounds for an expected committee size of 747 and $s_{\min} = 400$.}
\label{tab:num-example}
\begin{tabular}{rccc}
\toprule
 & & \multicolumn{2}{c}{Failure probability bound} \\
\cmidrule(lr){3-4}
$\mu$ & $N$ & Weighted sortition & Random sortition \\
\midrule
2{,}000   & 6{,}000   & $7.9 \times 10^{-13}$ & $2.9 \times 10^{-8}$ \\
10{,}000  & 30{,}000  & $3.0 \times 10^{-12}$ & $1.4 \times 10^{-7}$ \\
100{,}000 & 300{,}000 & $4.9 \times 10^{-12}$ & $2.0 \times 10^{-7}$ \\
\bottomrule
\end{tabular}
\end{table}

\section{Security Analysis of \systwo}
\label{apdx:sec-analysis}
Recall that a decentralized resource market operates a blockchain that manages and processes the market workload. A secure market means that it operates a secure blockchain that satisfies liveness and safety. This means that all transactions in the market are processed based on its network protocol, only valid blocks (containing valid transactions) are accepted, and that the system is live in the sense that it processes its workload and its blockchain is growing over time. We recall the definition of a secure ledger below.

\emph{Ledger security.} A ledger $\led$ is secure if it satisfies the following properties~\cite{garay2015bitcoin}:
\begin{description}
\item[Safety:] For any two time rounds $t_1$ and $t_2$ such that $t_1 \leq t_2$, and any two honest parties $P_1$ and $P_2$, the confirmed state of $\led$ (which includes all blocks buried under at least $k$ blocks, where $k$ is the depth parameter) maintained by $P_1$ at $t_1$ is a prefix of the confirmed state of $\led$ maintained by party $P_2$ at time $t_2$ with overwhelming probability.

\item[Liveness:] A valid transaction $tx$ broadcast at time round $t$ will be recorded on $\led$ at time at most $t + u$ with overwhelming probability, where $u$ is the liveness parameter.
\end{description}

Validity of market workload processing is governed by its network protocol. Thus, honest miners on the mainchain and the sidechain process the workload they receive using the original logic of the market protocol. To show that \systwo is secure, we have to prove that deploying \systwo does not violate the liveness and safety of the underlying market.

%Furthermore, we recall the failure probability of a committee and the failure probability of the autorecovery protocol of \sysone (note that this assumes random committee election).

\begin{corollary}[Failure of any committee]\label{cor:any-fail}
Consider $K$ modules, where module $m$ has $q_{i,m}=n_{i,m}/\mu$ for each class $i$.
Let $\overline B_m$, $\underline G_m$, and $g_{\min,m}$ be defined as in
Lemma~\ref{lem:sc-fail} using $q_{i,m}$, with $0<s_{\min}\le g_{\min,m}$ for all $m$.
The probability of the event that any \systwo's sidechain committee fails (denoted as $CF$) can be expressed as:

\[
\fitmath{
    \Pr\left[CF\right]
        \le \sum_{m=1}^{K}\left(
            \Pr[2\overline B_m\ge\underline G_m]
            +\exp\left(-\frac{g_{\min,m}}{2}
                \left(1-\frac{s_{\min}}{g_{\min,m}}\right)^2\right)\right).
    }
\]
\end{corollary}

\begin{proof}
Lemma~\ref{lem:sc-fail} applies to each module separately, with $q_i=q_{i,m}$.
The result follows from the union bound.
\end{proof}

\begin{lemma}[Autorecovery Failure]
\label{theorem-sec}
Consider a module $\module$, the probability of the event that the module's autorecovery fails (denoted as $AF_\module$) can be expressed as:
\[
\fitmath{
\Pr[AF_\module] \leq \left(\Pr[2\overline{B}\ge\underline{G}] +\exp\left(-\frac{\EE[\underline{G}]}{2} \left(1-\frac{s_{\min}}{\EE[\underline{G}]}\right)^2\right)\right)^{\kappa+1}}
\] 
\end{lemma}
\noindent where $\overline{B}$, $\underline{G}$, $s_{\min}$ are defined as in Appendix~\ref{apdx:class-comp}, and $\kappa$ is the number of backup committees.

\begin{proof}
    Consider a class-$i$ miner $j$, and a sidechain with $\kappa + 1$ (primary and backup) committees. We assume that Algorithm~\ref{alg:cap} is extended to randomly assign an elected miner to one of $\kappa + 1$ committees using an independent VRF-based sortition.

    Define the following events:
    
    \begin{itemize}
        \item $E_{j,e}$: miner $j$ is elected to participate in any committee, for any module in epoch $e$. This event happens with probability $(\kappa + 1)(n_{i,all}/\mu)$, uniformly across epochs.
        \item $A_{j,e, \module}$: miner $j$ is assigned to the module $\module$ for epoch $e$, conditional on being elected. This event happens with probability $n_i/n_{i,all}$, uniformly across epochs.
        \item $O_{j,e, l, \module}$: miner $j$ is assigned to $\module$'s $l^{th}$ committee. This event happens with probability $1/(\kappa+1)$, uniformly across epochs.    
    \end{itemize}

    Under weighted sortition, we define the event of being elected to  $\module$'s  $l^{th} \in \{0, \cdots, \kappa\}$ committee, represented by the indicator random variable $I_{j,e,l,\module}$. All three events need to happen. We obtain:
    \[
        \begin{aligned}
            \Pr[I_{j,e,l,\module}=1] &=\Pr[E_{j,e}\cap A_{j,e,\module} \cap O_{j,e, l, \module} ]\\
            &=\Pr[E_{j,e}] \\ 
            & \quad \cdot \Pr[A_{j,e,\module}\mid E_{j,e}] \\
            & \quad \cdot \Pr[O_{j,e, l, \module} |A_{j,e,\module} \cap E_{j,e}]\\
            &=\left((\kappa+1)\frac{n_{i,\mathrm{all}}}{\mu}\right)
            \left(\frac{n_i}{n_{i,\mathrm{all}}}\right) \frac{1}{(\kappa+1)}\\
            &=\frac{n_i}{\mu}=q_i.
        \end{aligned}
    \]
    Thus, for any module $\module$, epoch $e$, and committee $l$,
    \[
        G_{i,e,l}\sim\mathrm{Bin}(\mathscr{W}_{i,e},q_i),\qquad
        B_{i,e,l}\sim\mathrm{Bin}(\mathscr{M}_{i,e},q_i).
    \]
    We define
    \[
        G_{e,l}=\sum_{i = 1}^{C}G_{i,e,l},\qquad
        B_{e,l}=\sum_{i = 1}^{C} B_{i,e,l}.
    \] 
    and 
    \[
	\overline B =\sum_{i=1}^{C} \mathrm{Bin}(\overline{\mathscr{M}}_i,q_i),
    \qquad
	\underline G =\sum_{i=1}^{C}\mathrm{Bin}(\underline{\mathscr{W}}_i,q_i).
    \]

    Then $\underline{G}$ is stochastically dominated by $G_{e,l}$, and $\overline{B}$ stochastically dominates $B_{e,l}$.
    
    Committee $l$ fails, denoted $F_l$, if $2B_{e,l}\ge G_{e,l}$ or $G_{e,l}<s_{\min}$. and thus
    \[
        \Pr[F_l] \leq \Pr[2\overline{B}\ge\underline{G}] +\exp\left(-\frac{\EE[\underline{G}]}{2} \left(1-\frac{s_{\min}}{\EE[\underline{G}]}\right)^2\right).
    \]
     
     Autorecovery fails when all (primary and backup) committees fail, i.e., $\bigcap_{l=0}^{\kappa} F_l$. Thus, it remains to bound $\Pr[\bigcap_{l=0}^{\kappa}F_l]$, using negative association (NA).     
     
     Fix a miner $j$. Since $j$ is assigned to at most one committee of $\module$, the indicators $(I_{j,e,l,\module})_{l=0}^{\kappa}$ sum to at most one and are NA by the Zero-One Lemma~\cite{dubhashi1996balls}. The VRF outputs of distinct miners are independent, so by Proposition~7(1) of~\cite{dubhashi1996balls} the family $\{I_{j,e,l,\module}\}_{j,l}$ over all miners and all $l$ is NA. Now $G_{e,l}=\sum_i\sum_{j\in\mathscr W_{i,e}}I_{j,e,l,\module}$ and $B_{e,l}=\sum_i\sum_{j\in\mathscr M_{i,e}}I_{j,e,l,\module}$ are non-decreasing functions of pairwise disjoint subsets of this family, so by Proposition~7(2) of~\cite{dubhashi1996balls} the vectors $\mathbf G=(G_{e,l})_{l=0}^{\kappa}$ and $\mathbf B=(B_{e,l})_{l=0}^{\kappa}$ are each NA. Applying Proposition~7(2) again with the non-increasing maps $h_l(x)=-x$ on the singletons $\{l\}$, $-\mathbf G=(-G_{e,l})_{l=0}^{\kappa}$ is NA. Since $\mathbf B$ and $\mathbf G$ are determined by the VRF outputs of the malicious and honest miners respectively, and these sets are disjoint, $\mathbf B$ and $-\mathbf G$ are independent; by Proposition~7(1) the joint vector $(B_{e,0},\dots,B_{e,\kappa},-G_{e,0},\dots,-G_{e,\kappa})$ is NA.

    Let $f(b,g)=\mathbf 1\{2b+g\ge 0\}\vee\mathbf 1\{g>-s_{\min}\}$, so that $\mathbf 1\{F_l\}=f(B_{e,l},-G_{e,l})$. The function $f$ is non-decreasing: if $f(b_2,g_2)=1$ and $b_1\ge b_2$, $g_1\ge g_2$, then either $2b_1+g_1\ge 2b_2+g_2\ge 0$ or $g_1\ge g_2>-s_{\min}$, so $f(b_1,g_1)=1$. The pairs $(B_{e,l},-G_{e,l})$ occupy disjoint coordinates
    of the NA vector above, so the product inequality gives
    \[
        \begin{aligned}
            \Pr\Bigl[\bigcap_{l=0}^{\kappa}F_l\Bigr] 
            &=\EE\Bigl[\prod_{l=0}^{\kappa}f(B_{e,l},-G_{e,l})\Bigr]\\
            &\le\prod_{l=0}^{\kappa}\EE\bigl[f(B_{e,l},-G_{e,l})\bigr]\\
            &=\prod_{l=0}^{\kappa}\Pr[F_l].
        \end{aligned}
    \]
    
     Consequently, we have:
    \[
        \Pr[AF_\module] \leq \prod_{l=0}^{\kappa}\Pr[F_l] = (\Pr[F_l])^{\kappa + 1}
    \]
     
     which completes the proof.   
\end{proof}

\systwo has $K$ sidechains and autorecovery in \systwo fails if any of the autorecovery of any of its sidechains fails. Applying a union bound, we have the failure probability of \systwo's autorecovery, denoted as $\Pr[AF_{\systwo}]$, is $\Pr[AF_{\systwo}] \leq \sum_{m=1}^{K} \Pr[AF_m]$. 

\begin{theorem}
    Deploying \systwo on top of a secure decentralized resource market does not violate safety and liveness of this market.
\end{theorem}

We prove this theorem by proving two lemmas as follows.

\begin{lemma}[Preserving safety]
    \systwo preserves the safety of the underlying resource market
\end{lemma}

\begin{proof}
    Based on the mechanisms employed in \systwo design, we identify the following threats to safety:

    \begin{itemize}

        \item Service agreement violation, slacking and theft: As meta-blocks are pruned, servers and clients may try to violate service agreement terms, e.g., a server offers less service than agreed upon or a client tries to pay less than the agreed-upon price. Also, a server (who did not do the work) may claim it offered a service but the service delivery proofs have been pruned, or a client (who received the service) claims that no service has been delivered, as the meta-blocks containing the proofs have been pruned, to avoid paying the servers.  
        
        \item Dispute circumvention: A party keeps participating in the system under the same identity even if they were excluded or penalized because they misbehaved.
        
        \item Out-of-sync mainchain or invalid syncing: A sidechain committee does not issue a sync-transaction at the end of the epoch, causing the mainchain state to be out of sync. Or this committee issues a sync-transaction that contains invalid summaries, causing the mainchain to be synced using invalid records (thus violating mainchain quality).
        
        \item Producing invalid sidechain blocks: A sidechain committee produces invalid meta- or summary-blocks. This violates the sidechain quality, and in turn will lead to invalid mainchain syncing. 
    \end{itemize}

    \systwo is resilient to these attacks since it uses a secure PBFT-consensus run by committees elected in a non-biased way with a size that guarantees honest majority, and the use of a secure autorecovery protocol that ensures system recovery from interruption with overwhelming probability.

    \textit{Service agreement violation, slacking and theft} as well as \emph{dispute circumvention} are addressed as follows. The modules process their workload using the original logic of the underlying resource market. Meta-blocks containing the processed workload are not dropped until their corresponding sync-transaction is confirmed on the mainchain. Thus, anyone can verify the validity of the produced summaries based on the corresponding meta-blocks. Summary-blocks are permanent, where they contain summarized service agreements, service delivery proofs, and dispute verification/outcome. Also, payments are made to the servers on the mainchain based on the service records of these servers (i.e., clients do not control that). By the security of consensus, these summaries are guaranteed to be correct, and hence, any service claims that do not agree with these summaries will be rejected and corresponding malicious parties will be penalized. Our autorecovery protocol and the mass-syncing mechanism guarantee that the system will recover from any interruption, thus restoring its valid operation state.

    \textit{Out-of-sync mainchain, invalid syncing, or producing invalid sidechain blocks} are addressed by the security of our autorecovery protocol (along with mass-syncing), the security of consensus, and electing committees with a size that guarantees an honest majority. These ensure that committees follow the prescribed protocol, and thus produce valid summaries and sync-transactions to sync the mainchain based on correct records. The security of consensus guarantees that a committee with honest majority only agrees on valid blocks, with the autorecovery protocol taking care of cases when a committee fails by having backup committees step in.

    As a result, \systwo preserves the safety of the underlying resource market.  
\end{proof}

\begin{lemma}[Preserving liveness]
    \systwo preserves the liveness of the underlying resource market.
\end{lemma}

%\ghada{reviewer: The liveness proof (Lemma C.4) contradicts the inter-module recovery protocol: when a critical module is interrupted, dependent modules mine empty blocks, leaving valid transactions unprocessed. The paper's own experiments show recovery can exceed 35 minutes or fail entirely, yet the proof claims liveness preservation without adjusting the liveness paramete}

%\mo{The experiment showing the recovery can exceed 35 mins or failing entirely is when the weighted selection is applied the other way[for non critical modules]: the module is getting a high selection of lower-class miners ... e.g.; it's getting $40\%$ of C1 (higher class with lower adversarial mines), and $60\%$ of C2 (lower class with higher adversarial miners)  -- and we say in our experiment that this is bad, don't use it ... if you need to do something for low criticality modules, go with a balanced selection.}

\begin{proof}
    As \systwo employs committees to manage module sidechains, we identify the following threats to liveness:
    \begin{itemize}
        \item DoS attacks: A sidechain's committee excludes particular transactions from being published.
        \item Violating sidechain liveness: The committee managing important modules does not add blocks to its sidechain, or the leader does not issue a $\tx_\sync$. This will violate the liveness of this sidechain. Also, not adding blocks may cause other modules to mine empty blocks for an undetermined period.
        \item Violating public verifiability: This consists of all threats resulting from workload sharing and blockchain pruning on the public verifiability of the records on the side and main chains.
    \end{itemize}

    \emph{DoS attacks} are unfeasible since we employ a secure PBFT-based consensus with rotating epoch committees and leaders. Thus, it is infeasible for an attacker to hold full control of the transaction inclusion process on the long term.

    \textit{Violating sidechain liveness} is addressed by the use of a secure PBFT-based consensus and the autorecovery protocol of \systwo. Any disruption to liveness would be detected by the committee members (if the cause is the committee leader) and will be solved by a leader change process; or by having a backup committee step in (in case of an unresponsive committee). The loss of sync-transactions (due to rollbacks) or due to having a leader that does not issue them, is resolved by the mass-syncing process. By configuring the committee size and the autorecovery protocol parameters properly, system recovery is guaranteed with overwhelming probability.

    \textit{Violating public verifiability} is mitigated by the use of a secure PBFT-based consensus and our autorecovery protocol, which guarantee that only correct transactions are processed and published on the sidechains and that correct summaries are produced. Also, the mass-syncing mechanism will handle any interruptions that may happen to the syncing process, thus guaranteeing that no summaries will be lost. Hence, the mainchain will be synced based on the workload processed on the sidechain. Since meta-blocks are pruned after their corresponding $\tx_\sync$ is confirmed on the mainchain, anyone can verify that the (permanent) summaries are valid. By the security of consensus, it is guaranteed that only valid summaries are accepted and their correctness will not change once meta-blocks are pruned. All these factors preserve the public verifiability of market operation.  

    Thus, \systwo preserves the liveness of the underlying resource market.
\end{proof}

% that's all folks
\end{document}